\documentclass[a4paper,UKenglish,cleveref, autoref, thm-restate, pdfa]{lipics-v2021}
\usepackage[table]{xcolor}
\usepackage{booktabs,colortbl}
\usepackage{bm}
\usepackage{graphicx, tikz, svg}
\usepackage[ruled,linesnumbered,vlined]{algorithm2e}

\SetCommentSty{commentforalgorithmtwoe}
\SetKwComment{Comment}{$\triangleright$\ }{}
\crefname{algocf}{alg.}{algs.}
\Crefname{algocf}{Algorithm}{Algorithms}
\usetikzlibrary{matrix,calc,patterns,patterns.meta,fadings,backgrounds,intersections,through,arrows.meta}
\usepackage{awesomebox}
\usepackage{footnote}
\usepackage{pdfrender} 
\newcommand{\backgroundcontour}[1]{%
\textpdfrender{
	TextRenderingMode=FillStrokeClip,
	LineWidth=1.75pt,
	FillColor=white,
	StrokeColor=white,
	MiterLimit=1
}{#1}%
}

\definecolor{colorA}{HTML}{D55E00}
\definecolor{colorAdark}{HTML}{944000}

\definecolor{colorB}{HTML}{009E73}
\definecolor{colorBdark}{HTML}{006147}

\definecolor{colorC}{HTML}{0072B2}
\definecolor{colorCdark}{HTML}{005A8F}

\definecolor{colorD}{HTML}{F0E442}
\definecolor{colorDdark}{HTML}{5F5807}

\definecolor{colorE}{HTML}{E69F00}
\definecolor{colorEdark}{HTML}{755000}

\definecolor{colorF}{HTML}{56B4E9}
\definecolor{colorFdark}{HTML}{125C87}

\definecolor{colorG}{HTML}{CC79A7}
\definecolor{colorGdark}{HTML}{91366B}

\title{Revisiting \texorpdfstring{$\bm{O(n \log \log n)}$}{O(n log log n)} Chaining for Anchored Edit Distance}
\titlerunning{Revisiting \texorpdfstring{$\bm{O(n \log \log n)}$}{O(n log log n)} Chaining for Anchored Edit Distance}
\author{Nicola Rizzo}{University of Helsinki,
  Finland}{nicola.rizzo@helsinki.fi}{https://orcid.org/0000-0002-2035-6309}{European Union's Horizon Europe research and innovation programme under grant agreement No 101060011 (TeamPerMed).}
\author{Ragnar {Groot Koerkamp}}{Karlsruhe Institute of Technology, Germany}{ragnar.grootkoerkamp@gmail.com}{https://orcid.org/0000-0002-2091-1237}{}
\authorrunning{N. Rizzo and R. Groot Koerkamp} 
\Copyright{Nicola Rizzo and Ragnar Groot Koerkamp} 
\category{} 
\relatedversion{} 
\acknowledgements{The authors wish to thank the Finnish Computing Competence Infrastructure (FCCI) for supporting this project with computational and data storage resources.}
\nolinenumbers 
\supplementdetails[cite=llchain,swhid={swh:1:dir:78de8057c7f61afc5518ef3b8a46a5ac91e21d52}]{Software}{https://github.com/nrizzo/llchain}

\EventEditors{Nadia El-Mabrouk and Fabio Vandin}
\EventNoEds{2}
\EventLongTitle{26th International Conference on Algorithms for Bioinformatics (WABI 2026)}
\EventShortTitle{WABI 2026}
\EventAcronym{WABI}
\EventYear{2026}
\EventDate{August 31--September 2, 2026}
\EventLocation{L'Aquila, Italy}
\EventLogo{}
\SeriesVolume{390}
\ArticleNo{13}

\newcommand{\parag}[1]{\subparagraph{#1}}
\newcommand{\qbeg}{q_{\text{s}}}
\newcommand{\qend}{q_{\text{e}}}
\newcommand{\tbeg}{t_{\text{s}}}
\newcommand{\tend}{t_{\text{e}}}
\newcommand{\wprec}{\prec_\mathrm{w}}
\newcommand{\chainxprec}{\prec_{\mathtt{ChainX}}}
\newcommand{\sprec}{\prec_\mathrm{s}}

\newcommand{\startanchor}{a_\mathrm{start}}
\newcommand{\finalanchor}{a_\mathrm{end}}
\newcommand{\ddiag}{\Delta_\mathrm{diag}}
\newcommand{\rank}{\mathsf{rank}}
\DeclareMathOperator{\linfty}{L_\infty}
\DeclareMathOperator{\lone}{L_1}
\DeclareMathOperator{\anchorstart}{\mathsf{start}}
\DeclareMathOperator{\anchorend}{\mathsf{end}}
\DeclareMathOperator{\ov}{\mathsf{ov}}
\DeclareMathOperator{\cost}{\mathsf{cost}}
\DeclareMathOperator{\gcost}{\mathsf{gcost}}
\DeclareMathOperator{\pgcost}{\mathsf{pgcost}}
\DeclareMathOperator{\sgcost}{\mathsf{sgcost}}

\DeclareMathOperator{\connect}{\mathsf{connect}}
\DeclareMathOperator{\score}{\mathsf{score}}
\DeclareMathOperator{\fscore}{\mathsf{fscore}}

\DeclareMathOperator{\diag}{diag}

\DeclareMathOperator*{\argmin}{arg\,min}

\newtheorem{problem}[theorem]{Problem}
\crefname{observation}{observation}{observations}
\Crefname{observation}{Observation}{Observations}

\begin{document}
\maketitle

\begin{abstract}
    Colinear chaining is a classical heuristic for sequence alignment: it enables
    scalable genome comparison and is a main component of many state-of-the-art
    read mappers based on seed-chain-extend.
    The earliest $O(n \log \log n)$ and $O(n \log n)$ time algorithms by Eppstein et al.\ (J.\ ACM, 1992) chained $n$ fragments between two sequences $T$ and $Q$ while minimizing a gap cost based on the diagonal distance $\ddiag$ between consecutive fragments.
    They also forbid fragment overlaps, which are essential in current chaining formulations: in long-read mapping, overlaps improve sensitivity and avoid restrictions on the fragment class considered.
    Jain, Gibney, and Thankachan (J.\ Comput.\ Biol.\ 2022) recently combined a
    $\Delta_\text{diag}=|\Delta_T-\Delta_Q|$ overlap cost with the classic $\linfty=\max(\Delta_T, \Delta_Q)$ gap cost that
    takes the maximum between the horizontal and vertical gap between the fragments and they proved that chaining under this cost model is equivalent to the \emph{anchored edit distance}.

    We improve the existing $O(n \log^3 n)$-time algorithm for anchored edit
    distance to $O(n \log \log n)$ time in $O(n)$ space, by combining the gap-cost computation of Chao and Miller (Algorithmica, 1995) with the
    overlap-cost computation of Baker and Giancarlo (ESA, 1998).
    By developing
    \texttt{llchain}, a simpler $O(n \log n)$-time
    implementation of our method, we show how chaining algorithms
    that might have been recently overlooked by the bioinformatics community scale competitively to millions of fragments and large genomes.
    On average, \texttt{llchain} is $10\times$ faster than other methods on instances
    with 3\,000\,000 anchors, and over $2.3\times$ faster on MEMs
    between HiFi reads and a reference human genome.
\end{abstract}

\section{Introduction}
Colinear chaining, known throughout the decades as \emph{Wilbur--Lipman sequence
comparison}, \emph{global fragment chaining}, \emph{sparse dynamic programming},
and just \emph{chaining} is a fundamental technique for approximate alignment that dates back to the 1980s~\cite{wilbur1983rapid,wilbur1984context,DBLP:books/cu/Gusfield1997,DBLP:books/oldenbusch/2013}.
Chaining techniques have been widely adopted in different sub-fields of computational genomics: whole-genome alignment tools~\cite{ohlebusch2006chaining}; sequence aligners~\cite{kurtz2004versatile,DBLP:journals/bioinformatics/Li18,DBLP:journals/bioinformatics/KoerkampI24}; read mappers employing the seed-chain-extend strategy~\cite{sahlin2023survey}; and even assembly evaluation~\cite{DBLP:journals/bioinformatics/BushmanovaALSP16,DBLP:journals/bioinformatics/MikheenkoPSAG18}.
While the general chaining problem statement is consistent in the literature --
to find the highest-scoring ordered subset of the input fragments that is
consistent with an alignment -- the exact formulation and algorithms change depending on the application, and it can be difficult to compare the different results.
In particular, the behaviour of chaining is mainly affected by fragment score, gap cost, and overlap cost.
Popular aligners use a heuristic chaining cost based on the diagonal distance and allow overlaps~\cite{DBLP:journals/bioinformatics/Li18,DBLP:journals/ploscb/RenC21,lastz}, whereas some methods chain without overlaps and consider a function of the gaps in the two sequences~\cite{abouelhoda2004chainer,DBLP:journals/almob/OttoHGS11}, in one case obtaining a \emph{lower bound} on the cost of the optimal alignment \cite{DBLP:journals/bioinformatics/KoerkampI24}.
Instead, we consider the \emph{anchored edit
distance}~\cite{BakerG02,DBLP:conf/cpm/MakinenS20,JainGT22}, which gives an \emph{upper} bound
on the true edit distance. In this model, a number of \emph{anchors} or
\emph{fragments} are given, which are exact matches between the two input
sequences. Then, the anchored edit distance is the minimum cost of a global alignment where
only the matches supported by an anchor are free.

\begin{definition}[Anchored edit distance]\label{def:aed}
  Given two strings $T$ and $Q$, an \emph{anchor} (or \emph{fragment}) $a=(t_s, t_e, q_s, q_e)$ is a
  tuple for which $T[t_s, t_e) = Q[q_s, q_e)$.
  Given an arbitrary set of $n$ anchors, the \emph{anchored edit distance} is the cost of
  a global alignment of $T$ and $Q$ that minimizes the total number of insertions,
  deletions, substitutions, and matches \emph{not supported by an anchor}.
  Thus, only \emph{supported} matches are free.
\end{definition}

Jain et al.~\cite{JainGT22} show that the anchored edit distance can be found by computing a minimum-cost chain using an $O(n \log^3 n)$ algorithm,
and Rizzo et al.~\cite{RizzoCM25} simplified the chaining cost as follows.

\begin{theorem}[Equivalence to chaining]
  The anchored edit distance equals the minimum-cost chain of fragments under
  the following colinear chaining cost between the end $E=(t_e, q_e)$ of one
  anchor $a$ and the
  start $S'=(t'_s, q'_s)$ of the next anchor $a'$:
  \begin{equation*}
  \connect(a, a') := \max(\linfty(E, S'), \ddiag(E, S')),
  \end{equation*}
  where
  \begin{align*}
    \linfty((t, q), (t', q')) &:= \max(0,\Delta_T, \Delta_Q)= \max(0,t'-t, q'-q) \text{\quad is the maximum gap, and} \\
    \ddiag((t, q), (t', q')) &:= |\Delta_T - \Delta_Q| = \lvert\diag(t,q) - \diag(t',q') \rvert \text{\quad is the diagonal distance}.
  \end{align*}
\end{theorem}

\parag{Contributions.}
Our main contribution is to connect the anchored edit distance to previously
studied chaining problems, resulting in an $O(n\log\log n)$ theoretical
algorithm working in $O(n)$ space.
Our method works by combining the $\linfty$ gap-cost computation in $O(n \log \log n)$ time of Eppstein et al.~\cite{EppsteinGGI92one} and Chao and Miller~\cite{DBLP:journals/algorithmica/ChaoM95} with the $\ddiag$ overlap cost computation of Baker and Giancarlo~\cite{BakerG02}.
The $O(n \log \log n)$ time complexity is obtained by using a fast integer sorting algorithm as well as predecessor structures.
A simplified $O(n \log n)$ time implementation, \texttt{llchain}, performs up to $10\times$ faster ($2.3\times$ on average) than the existing solutions for anchored edit distance, when chaining the maximal exact matches (MEMs) of length 50 or more between HiFi reads and a human genome.

\section{Related work}\label{sec:relatedwork}

\enlargethispage{1em}
\interfootnotelinepenalty=10000

To provide the needed historical context to our work and to help the reader navigate the literature, we start with a review of the most important results on pairwise chaining focusing on the time complexity in the number of input anchors, summarized in \Cref{tab:chaining}.

The original chaining formulations by Wilbur and Lipman~\cite{wilbur1984context}
considered as input exact-match fragments between two sequences $T$ and $Q$, to
be selected and connected together in an order that is consistent with an
alignment (i.e.\ a \emph{chain}) while maximizing the number of matched bases
minus a constant cost for each gap~\cite{wilbur1983rapid}, or maximizing \emph{weighted} matches minus an arbitrary gap
cost $g(\Delta_T, \Delta_Q)$~\cite{wilbur1984context}, where $\Delta_Q$ and
$\Delta_T$ are the lengths of the gaps in $T$ and $Q$ (see \Cref{fig:distances}). They
provided an $O(n^2)$ time algorithm.
In 1992, the first efficient sub-quadratic-time algorithms were developed by Eppstein et
al.~\cite{EppsteinGGI92one,EppsteinGGI92two} to maximize the number of matched bases minus gap costs computed as $f(\Delta_\text{diag})$, a one-dimensional function in the diagonal distance between consecutive fragments, and they run in $O(n \log \log n)$ or $O(n \log n)$ time, depending on whether $f$ is linear or concave.
Gap costs of the form $f(\Delta_\text{diag})$ are still widely used today~\cite{DBLP:journals/bioinformatics/Li18,DBLP:journals/ploscb/RenC21} but they do not explicitly penalize large gaps.
This was addressed in efficient algorithms in 1995\footnote{%
Zhang et al.~\cite{DBLP:journals/jcb/ZhangRHM94} define an affine gap score for chaining fragments between $k \ge 2$ sequences already in 1994, but they address the problem with $k$-dimensional trees and do not provide a worst-case complexity analysis.} by Myers and Miller~\cite{DBLP:conf/soda/MyersM95} when generalizing chaining to two or more sequences, and by Chao and Miller~\cite{DBLP:journals/algorithmica/ChaoM95} when extending the approach by Eppstein et al.\ to find $m$ best-scoring local alignments.
Myers and Miller's gap cost $\varepsilon \cdot \min(\Delta_T, \Delta_Q) +
\lambda \cdot \ddiag$
found a lot of attention in the alignment of whole
genomes~\cite{ohlebusch2006chaining}, whereas pairwise chaining implementations
like \texttt{CHAINER}~\cite{abouelhoda2004chainer} concentrated on the $\lone =
\Delta_T + \Delta_Q$ gap cost (i.e. $(\varepsilon,\lambda)=(2,1)$).
To the best of our knowledge, the first chaining implementation to support
generic $\varepsilon$ and $\lambda$ -- and thus $\linfty = \max(\Delta_T,
\Delta_Q)$ for $(\varepsilon, \lambda) = (1,1)$ -- is \texttt{clasp}~\cite{DBLP:journals/almob/OttoHGS11}.
All chaining formulations before 1998, and many formulations since, do not allow
anchors to overlap (see e.g.\ \Cref{fig:chain}): historically overlaps have been ignored, removed in a pre-processing step, or recovered in downstream refinement phases~\cite{DBLP:journals/bioinformatics/ChaoZOM95,DBLP:journals/bib/ChainKOS03}.
However, overlaps improve sensitivity and are widely employed in long-read
alignment, where the chaining formulation of
\texttt{minimap2}~\cite{DBLP:journals/bioinformatics/Li18} is the \emph{de
facto} standard way of chaining minimizer seeds, a popular class of sampled fixed-length seeds~\cite{sahlin2023survey}.
Other notable applications of chaining include the evaluation of genome assemblies~\cite{DBLP:journals/bioinformatics/BushmanovaALSP16,DBLP:journals/bioinformatics/MikheenkoPSAG18}, chaining with inversions~\cite{DBLP:journals/ploscb/RenC21}, and computing an admissible lower bound for the A*PA heuristic to sequence alignment~\cite{DBLP:journals/bioinformatics/KoerkampI24}.

\begin{figure}[t]
    \centering
    \includegraphics[]{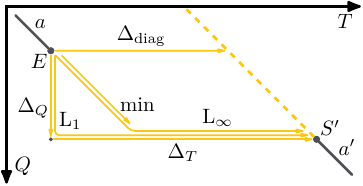}
    \caption{Different possible chaining costs between two anchors $a$ and $a'$:
    the distance is between the end $E$ of $a$ and start $S'$ of $a'$, and can
    be the distance in $Q$ ($\Delta_Q$), the distance in $T$ ($\Delta_T$), or the
    sum ($\lone$), minimum ($\min$), maximum ($\linfty$), or absolute difference
    ($\ddiag=\lvert \Delta_T-\Delta_Q \rvert$) of these.}\label{fig:distances}
\end{figure}

\enlargethispage{1em}

A complementary problem is that of \emph{anchored longest common subsequence} (anchored LCS), where
one searches an alignment that maximizes the number of matches supported by an anchor.
Baker and Giancarlo~\cite{BakerG02} showed in 1998 that chaining to minimize an $\lone=\Delta_T+\Delta_Q$ gap cost and a $\Delta_\text{diag}$ overlap cost actually solves the anchored LCS, and solved the problem in $O(n \log \log n)$ time.
The connection to LCS was rediscovered in 2021 by M\"akinen and 
Sahlin~\cite{DBLP:conf/cpm/MakinenS20} with an alternative but equivalent chaining formulation.
Further work in this direction includes LCS$k$ \cite{lcsk,lcsk-refined}, where anchors are the
set of all $k$-mers (exact matches of length $k$) and may not overlap, and LCS$k+$ \cite{lcsk++}, where
anchors are all matches of length at least $k$. In \cite{PaveticKMZS18}, an $O(n\log\ell)$ algorithm is given, where $\ell\leq n$ is the number of anchors in the optimal chain.

\begin{table}[h]
    \caption{Main algorithms and implementations (in \texttt{typewriter font}) for the global or semi-global chaining of $n$ fragments between two sequences $T$ and $Q$, with $Q$ the smaller string. For fragment types, ``exact'' refers to exact matches, ``square'' to matches of the same length, and ``inexact'' to approximate matches.
    For ``arbitrary'' fragment scores, each fragment has a score in input, whereas ``matches'' is equal to the number of matching bases in the corresponding alignment (also known as mutual coverage~\cite{DBLP:conf/cpm/MakinenS20}).
    For the gap costs, $f(\Delta_T,\Delta_Q)$ is an arbitrary 2D function, $\ell(\Delta_\text{diag})$ is a linear 1D, $f_\cup$ is a concave function, $p_\cup$ is a piece-wise function with a concave part flanked by constant values, $\varepsilon$ and $\lambda$ are linear substitution and indel costs, and $\linfty = \max(\Delta_T,\Delta_Q)$. Gap costs that have $\linfty$ as a special case are marked in yellow, and $\ddiag$ overlap costs in blue.}
    \label{tab:chaining}
    \footnotesize\centering\makebox[\textwidth][c]{%
    \begin{tabular}{ccccccccccc}
        \toprule
         chaining formulation & fragment type & fragment score & gap cost & overlap cost & time complexity \\
         \midrule
         Wilbur and Lipman, 1983~\cite{wilbur1983rapid} & exact & length & constant & -- & $O(n^2)$ \\
         \rowcolor{lipicsLightGray}
         Wilbur and Lipman, 1984~\cite{wilbur1984context} & square & context-dependent & \cellcolor{lipicsYellow!75!lipicsLightGray}$f(\Delta_T,\Delta_Q) \ge 0$ & -- & $O(n^2)$ \\
         Eppstein et al., 1992~\cite{EppsteinGGI92one} & exact & matches & $\ell(\Delta_\text{diag})$ & $\Delta_\text{diag} = 0$ penalty & $O(n \log \log n)$ \\
         \rowcolor{lipicsLightGray}
         Eppstein et al., 1992~\cite{EppsteinGGI92two} & exact & matches & $f_\cup(\Delta_\text{diag})$ & $\Delta_\text{diag} = 0$ penalty & $O(n \log n)$ \\
          & &  & \cellcolor{lipicsYellow!75}$\varepsilon\cdot\min(\Delta_T,\Delta_Q) \;+$  &  & \\
         \multirow{-2}{*}{Myers and Miller, 1995~\cite{DBLP:conf/soda/MyersM95}} & \multirow{-2}{*}{inexact} & \multirow{-2}{*}{arbitrary} & \cellcolor{lipicsYellow!75}$+ \lambda \cdot |\Delta_T - \Delta_Q|$ & \multirow{-2}{*}{--} & \multirow{-2}{*}{$O(n \log^2 n)$} \\
         \rowcolor{lipicsLightGray}
         Chao and Miller, 1995  \cite{DBLP:journals/algorithmica/ChaoM95} & exact & length & \cellcolor{lipicsYellow!75!lipicsLightGray}affine & -- & $O(n \log \log n)$ \\
         Baker and Giancarlo, 1998~\cite{BakerG02} & exact & -- & $L_1 = \Delta_T + \Delta_Q$ & \cellcolor{-lipicsYellow!50!black}$\Delta_{\text{diag}}$ & $O(n \log \log n)$\\
         \rowcolor{lipicsLightGray}
         Abouelhoda and Ohlebusch,  &  &  &  &  &  \\
         \rowcolor{lipicsLightGray}2004~\cite{abouelhoda2004chainer,DBLP:journals/jda/AbouelhodaO05} \texttt{CHAINER} & \multirow{-2}{*}{inexact} & \multirow{-2}{*}{arbitrary} & \multirow{-2}{*}{$L_1$} & \multirow{-2}{*}{--} & \multirow{-2}{*}{$O(n \log n)$} \\
         Abouelhoda and Ohlebusch, & & & \cellcolor{lipicsYellow!75}$\varepsilon \cdot \min(\Delta_T, \Delta_Q) \;+$ & & \\
         2005~\cite{DBLP:journals/jda/AbouelhodaO05} & \multirow{-2}{*}{inexact} & \multirow{-2}{*}{arbitrary} & \cellcolor{lipicsYellow!75}$+ \lambda \cdot |\Delta_T - \Delta_Q|$ & \multirow{-2}{*}{--} & \multirow{-2}{*}{$O(n \log n \log \log n)$} \\
         \rowcolor{lipicsLightGray}
         Otto et al., 2011~\cite{DBLP:journals/almob/OttoHGS11} & & & \cellcolor{lipicsYellow!75!lipicsLightGray}$\varepsilon\cdot\min(\Delta_T,\Delta_Q) \;+$ & & \\
         \rowcolor{lipicsLightGray} \texttt{clasp} & \multirow{-2}{*}{inexact} & \multirow{-2}{*}{arbitrary} & \cellcolor{lipicsYellow!75!lipicsLightGray}$+ \lambda \cdot |\Delta_T - \Delta_Q|$ & \multirow{-2}{*}{--} & \multirow{-2}{*}{$O(n \log^2  n)$} \\
         Bushmanova et al., 2016~\cite{DBLP:journals/bioinformatics/BushmanovaALSP16} & & & & & \\
         \texttt{rnaQUAST} & \multirow{-2}{*}{inexact} & \multirow{-2}{*}{arbitrary} & \multirow{-2}{*}{--} & \multirow{-2}{*}{$\ov_T$} & \multirow{-2}{*}{$O(n^2)$} \\
         \rowcolor{lipicsLightGray}%
         Li, 2018~\cite{DBLP:journals/bioinformatics/Li18} & fixed-length, & & & & \\
         \rowcolor{lipicsLightGray}%
         \texttt{minimap2} & exact & \multirow{-2}{*}{matches} & \multirow{-2}{*}{$p_\cup(\Delta_\text{diag})$} & \multirow{-2}{*}{$p_\cup(\Delta_\text{diag})$} & \multirow{-2}{*}{$O(n^2)$} \\
         Sahlin and M\"akinen, 2020~\cite{DBLP:conf/cpm/MakinenS20} & exact & matches & 0 & 0 & $O(n \log n)$ \\
         \rowcolor{lipicsLightGray}%
         Ren and Chaisson, 2021~\cite{DBLP:journals/ploscb/RenC21} & & & & & \\
         \rowcolor{lipicsLightGray}%
         \texttt{lra} & \multirow{-2}{*}{inexact} & \multirow{-2}{*}{arbitrary} & \multirow{-2}{*}{$f_\cup(\Delta_\text{diag})$} & \multirow{-2}{*}{--} & \multirow{-2}{*}{$O(n \log^2 n)$} \\
         Sahlin and M\"akinen, & & mutual, weighted & & & \\
         2021~\cite{DBLP:conf/cpm/MakinenS20,DBLP:journals/bioinformatics/SahlinM21} \texttt{uLTRA} & \multirow{-2}{*}{inexact} & coverage & \multirow{-2}{*}{$\Delta_Q$} & \multirow{-2}{*}{$\varepsilon \cdot \max(\text{ov}_T, \text{ov}_Q)$} & \multirow{-2}{*}{$O(n^2)$} \\
         \rowcolor{lipicsLightGray}%
         Jain et al., 2022~\cite{JainGT22,RizzoCM25} & exact & -- & \cellcolor{lipicsYellow!75!lipicsLightGray} $\linfty$ & \cellcolor{-lipicsYellow!70!black}$\Delta_\text{diag}$ & $O(n \log^3 n)$ \\
         Groot Koerkamp and Ivanov, & non-overlapping, & & & & \\
         2024~\cite{DBLP:journals/bioinformatics/KoerkampI24} \texttt{A*PA} & inexact & \multirow{-2}{*}{edit distance} & \multirow{-2}{*}{$\max(\Delta_\text{diag},\lfloor\Delta_T/k\rfloor)$} & \multirow{-2}{*}{--} & \multirow{-2}{*}{$O(n \log^2 n)$} \\
         \rowcolor{lipicsLightGray}%
         \textbf{this work},  Alg 1 & exact & -- & \cellcolor{lipicsYellow!75!lipicsLightGray} $\linfty$ & \cellcolor{-lipicsYellow!70!black} $\Delta_\text{diag}$ & $O(n \log \log n)$ \\
         \textbf{this work}, \texttt{llchain} & exact & -- & \cellcolor{lipicsYellow!75} $\linfty$ & \cellcolor{-lipicsYellow!50!black} $\Delta_\text{diag}$ & $O(n \log n)$ \\
         \bottomrule
    \end{tabular}}
\end{table}

\section{Preliminaries on chaining, range minimum queries, priority queues}
\label{sec:preliminaries}
\enlargethispage{1em}
We adapt the notation from Rizzo, C\'aceres, and M\"akinen~\cite{RizzoCM25} to right-open intervals and 0-indexed strings.
We denote integer interval $\lbrace x, x+1, \dots, y \rbrace$ as $[x..y]$, or
$[x..y+1)$.
  Given $0$-indexed strings $T$ and $Q$ over finite alphabet $\Sigma$, interval pair $a = ([\tbeg..\tend), [\qbeg..\qend))$ is an \emph{exact match anchor} (or fragment, or just anchor) between $T$ and $Q$ if $T[\tbeg..\tend) = Q[\qbeg..\qend)$,
where $0\leq \tbeg \leq \tend \leq |T|$ and $0\leq \qbeg \leq \qend \leq |Q|$
(empty anchors are used as sentinels at the start and end of the text).
The main invariant of exact match anchors that we will sometimes use is $\tend - \tbeg = \qend - \qbeg \iff \tend - \qend = \tbeg - \qbeg$.

A colinear chaining formulation is first defined via \emph{anchor precedence}, a partial order $\prec$ specifying which anchors can follow one another to compose a chain $a_1 \prec \dots\prec a_c$.
One then maximizes the \emph{score} of a chain.
Indeed, all chaining formulations in \Cref{tab:chaining} maximize a formula of the type
$
    \score(A) = \sum_{i=1}^{c} \fscore(a_i) - \sum_{i=1}^{c-1} \connect(a_i, a_{i+1}),
$
where $\fscore(a)$ is the non-negative fragment score obtained by picking anchor $a$, and $\connect(a_i,a_{i+1})$ is the non-negative penalty (cost) of choosing $a_{i+1}$ after $a_i$.
Since our chaining formulation will have no fragment score, our goal will be to minimize a function which is very similar to $\cost(A) = \sum_{i=1}^{c-1} \connect(a_i, a_{i+1})$, the cost of connecting the anchors in a chain.

\subsection{Full problem statement}
Consider the following definitions alongside \Cref{fig:chain}.

\begin{definition}[Anchor precedence and colinear chain~\cite{BakerG02,JainGT22}]\label{def:precedence}
  Given $a = ([\tbeg..\tend), [\qbeg..\qend))$ between $T$ and $Q$, let $\anchorstart(a) = (\tbeg,\qbeg)$ be the start-point of $a$.
  Then $a$ \emph{precedes} $a' = ([\tbeg'..\tend'),[\qbeg'..\qend'))$, in symbols $a \prec a'$, if $\anchorstart(a)$ is distinct from and reaches $\anchorstart(a')$ in a possible alignment, that is,
    \[
    a \prec a'
    \quad\text{iff}\quad
        \tbeg \le \tbeg' \text{ and }
        \qbeg \le \qbeg' \text{ and }
        (\tbeg, \qbeg) \neq (\tbeg', \qbeg').
    \]
    Then, a sequence of anchors $a_1, \dots, a_c$ is called a \emph{colinear chain} or just $\prec$-chain if $a_i \prec a_{i+1}$ for all $i \in [1..c)$.
    Note that relation $\prec$ is referred to as \emph{weak precedence} in \cite{JainGT22}.
\end{definition}

\begin{definition}[{Connect function~\cite{JainGT22}, \cite[Lemma 1]{RizzoCM25}}]
    \label{def:connect}
    Given anchors $a \prec a'$, we define the following quantities (some shown in \cref{fig:distances}):
    \begin{align*}
\Delta_T(a,a') &:= \tbeg' - \tend \quad\text{and}\quad \Delta_Q(a,a') := \qbeg' - \qend  &\text{$T$-gap and $Q$-gap}\\
\linfty(a,a') &:= \max \big(0, \Delta_T(a, a'), \Delta_Q(a, a') \big) &\text{max-gap}\\
\ov_T(a, a') &:= \tend - \tbeg' \quad\text{and}\quad \ov_Q(a, a') := \qend - \qbeg' & \text{$T$-overlap and $Q$-overlap}\\
\ddiag(a, a') &:= |\Delta_T(a,a') - \Delta_Q(a, a')| = |\diag(a) - \diag(a')| & \text{diagonal distance}
    \end{align*}
where $\diag(a) = \tbeg - \qbeg$ is the diagonal of $a$.
    Then the cost of connecting two anchors is obtained as the maximum of the $\linfty$ gap cost and the $\Delta_{\diag}$ overlap cost:
    \[
        \connect(a,a') = \max \!\big( \!\linfty(a,a'), \ddiag(a,a') \big) = \begin{cases}
            \linfty(a, a') &\text{if}\;(\tend,\qend)\preceq (\tbeg',\qbeg') \;\text{(gap-gap)},\\
            \ddiag(a,a') &\text{otherwise} \;\text{(overlap case).}
        \end{cases}
    \]
    \end{definition}

    \begin{definition}[Global and semi-global cost~\cite{JainGT22}]\label{def:cost}
    Let $A = a_1, \dots, a_c$ be a $\prec$-chain of anchors between $T$ and $Q$.
    We define the \emph{global (alignment) cost} of $A$ as $\gcost(A) = \sum_{i=0}^{c} \connect(a_i, a_{i+1})$,
    where dummy anchors $a_0 = \startanchor = ([0,0),[0,0))$ and $a_{c+1} = \finalanchor = ([\lvert T \rvert, \lvert T \rvert), [\lvert Q \rvert, \lvert Q \rvert))$ indicate the start and end of the alignment.
    The \emph{semi-global cost} $\sgcost(A)$ ignores the initial and final gap in $T$, that is,
    \[
        \sgcost(A) = \Delta_Q(\startanchor,a_1) + \bigg(\sum_{i=1}^{c-1} \connect(a_{i}, a_{i+1}) \bigg) + \Delta_Q(a_c,\finalanchor).
    \]
\end{definition}

\begin{figure}[t]
    \quad
    \begin{minipage}[c]{0.65\linewidth}
        \raggedright
            \begin{tikzpicture}[remember picture]
\newcommand{\stripedunderline}[2]{%
\tikz\node[inner sep=2pt,minimum height=.75ex,minimum width=2ex,rectangle,overlay,shift={(-.5ex,-1.25ex)},fill=#2] (#1) {};%
}
\newcommand{\solidunderline}[2]{%
\tikz\node[inner sep=2pt,minimum height=.75ex,minimum width=2ex,rectangle,overlay,shift={(-.5ex,-1.25ex)},pattern={Lines[angle=45,distance=2pt,line width=1pt]},pattern color=#2] (#1) {};%
}
    \tikzset{
		invisible/.style = {minimum size=0pt, inner sep=0pt}
	}
\matrix[
	matrix of nodes,
	minimum size=2mm,
	column sep={3.5mm,between origins},
	row sep={3mm,between origins},
	inner sep=0pt,
	nodes={anchor=base},
	anchor=west,
	font=\ttfamily,
] (Q) {
A&C&A&T&C&T&G&C&C&A&A&C&A&T&A&T&C&C\\
};

\matrix[
	matrix of nodes,
	minimum size=2mm,
	column sep={3.5mm,between origins},
	row sep={3mm,between origins},
	inner sep=0pt,
	nodes={anchor=base},
	anchor=west,
	font=\ttfamily,
] (T) at (0,-1cm) {
A&C&A&T&C&C&G&G&C&C&A&T&A&T&A&T&C&C\\
};

\draw[-,very thick, colorA] ($(Q-1-2.north west) + (0,.05cm)$) -- node[invisible] (Qanchor0) {} ($(Q-1-5.north east) + (0,.05cm)$);
\draw[-,very thick, colorA] ($(T-1-2.north west) + (0,.05cm)$) -- node[invisible] (Tanchor0) {} ($(T-1-5.north east) + (0,.05cm)$);
\draw[-,very thick, colorB] ($(Q-1-7.south west) + (0,-.05cm)$) -- node[invisible] (Qanchor1) {} ($(Q-1-10.south east) + (0,-.05cm)$);
\draw[-,very thick, colorB] ($(T-1-8.south west) + (0,-.05cm)$) -- node[invisible] (Tanchor1) {} ($(T-1-11.south east) + (0,-.05cm)$);
\draw[-,very thick, colorC] ($(Q-1-12.north west) + (0,.05cm)$) -- node[invisible] (Qanchor2) {} ($(Q-1-16.north east) + (0,.05cm)$);
\draw[-,very thick, colorC] ($(T-1-10.north west) + (0,.05cm)$) -- node[invisible] (Tanchor2) {} ($(T-1-14.north east) + (0,.05cm)$);
\draw[-,very thick, colorD] ($(Q-1-13.south west) + (0,-.05cm)$) -- node[invisible] (Qanchor3) {} ($(Q-1-17.south east) + (0,-.05cm)$);
\draw[-,very thick, colorD] ($(T-1-13.south west) + (0,-.05cm)$) -- node[invisible] (Tanchor3) {} ($(T-1-17.south east) + (0,-.05cm)$);
\draw[-,colorA] (Qanchor0) to node[pos=0.7] (a1) {\backgroundcontour{$a_1$}} (Tanchor0);
\node[colorAdark] at (a1) {$a_1$};
\draw[-,colorB] (Qanchor1) to node[pos=0.3] (a2) {\backgroundcontour{$a_2$}} (Tanchor1);
\node[colorBdark] at (a2) {$a_2$};
\draw[-,colorC] (Qanchor2) to node[pos=0.7] (a3) {\backgroundcontour{$a_3$}} (Tanchor2);
\node[colorCdark] at (a3) {$a_3$};
\draw[-,colorD] (Qanchor3) to node[pos=0.3] (a4) {\backgroundcontour{$a_4$}} (Tanchor3);
\node[colorDdark] at (a4) {$a_4$};
\begin{scope}[on background layer]
\draw[draw=none,fill=colorA!50] ($(Q-1-1.north west)$) rectangle ($(Q-1-1.south east)!0.5!(Q-1-2.south west)$);
\draw[draw=none,fill=colorA!50] ($(T-1-1.north west)$) rectangle ($(T-1-1.south east)!0.5!(T-1-2.south west)$);
\draw[draw=none,fill=colorB!50] ($(Q-1-6.north west)!0.5!(Q-1-5.north east)$) rectangle ($(Q-1-6.south east)!0.5!(Q-1-7.south west)$);
\draw[draw=none,fill=colorB!50] ($(T-1-6.north west)!0.5!(T-1-5.north east)$) rectangle ($(T-1-7.south east)!0.5!(T-1-8.south west)$);
\draw[draw=none,fill=colorC!50] ($(Q-1-11.north west)!0.5!(Q-1-10.north east)$) rectangle ($(Q-1-11.south east)!0.5!(Q-1-12.south west)$);
\draw[draw=none,pattern={Lines[angle=45,distance=2pt,line width=1pt]},pattern color=colorC!50] ($(T-1-10.north west)!0.5!(T-1-9.north east)$) rectangle ($ (T-1-12.south east)!0.5!(T-1-11.south west) $);
\draw[draw=none,pattern={Lines[angle=45,distance=2pt,line width=1pt]},pattern color=colorD!50] ($(Q-1-13.north west)!0.5!(Q-1-12.north east)$) rectangle ($ (Q-1-17.south east)!0.5!(Q-1-16.south west) $);
\draw[draw=none,pattern={Lines[angle=45,distance=2pt,line width=1pt]},pattern color=colorD!50] ($(T-1-13.north west)!0.5!(T-1-12.north east)$) rectangle ($ (T-1-15.south east)!0.5!(T-1-14.south west) $);
\draw[draw=none,fill=lipicsLightGray] ($(Q-1-17.north east)!0.5!(Q-1-18.north west)$) rectangle ($(Q-1-18.south east)$);
\draw[draw=none,fill=lipicsLightGray] ($(T-1-17.north east)!0.5!(T-1-18.north west)$) rectangle ($(T-1-18.south east)$);
\end{scope}
    \node[anchor=west] at (0,-2cm) (cost) {$\max(\stripedunderline{u1}{colorA!50} 1,\, \stripedunderline{u2}{colorA!50} 1) + \max(\stripedunderline{u3}{colorB!50} 1, \,\stripedunderline{u4}{colorB!50} 2) + \lvert \stripedunderline{u5}{colorC!50} 1 - (-\solidunderline{u6}{colorC!50} 2) \rvert \;+$};
    \node[anchor=west] at (0,-2.75cm) {$\quad+\; \lvert - \solidunderline{u7}{colorD!50} 4 - (-\solidunderline{u8}{colorD!50} 2) \rvert + \max(\stripedunderline{u9}{lipicsLightGray} 1, \, \stripedunderline{u10}{lipicsLightGray} 1) = 9$};
	\node[left] at (Q.west) {$T=$};
	\node[left] at (T.west) {$Q=$};
	\node[left] at (cost.west) {$\gcost(A)=$};

\foreach \i[evaluate=\i as \ii using int(\i+1)] in {1,5,6,10,11,12,16,17}
{
	\node[anchor=base] at ($ (Q-1-\ii.north) + (0,2ex) $) {\scriptsize \i};
}
\foreach \i[evaluate=\i as \ii using int(\i+1)] in {1,5,7,9,11,12,14,17}
{
	\node[anchor=base] at ($ (T-1-\ii.south) - (0,2.5ex) $) {\scriptsize \i};
}
\end{tikzpicture}%
    \end{minipage}
    \hfill
    \begin{minipage}[c]{0.25\linewidth}
        \raggedleft
        \includesvg[width=\linewidth]{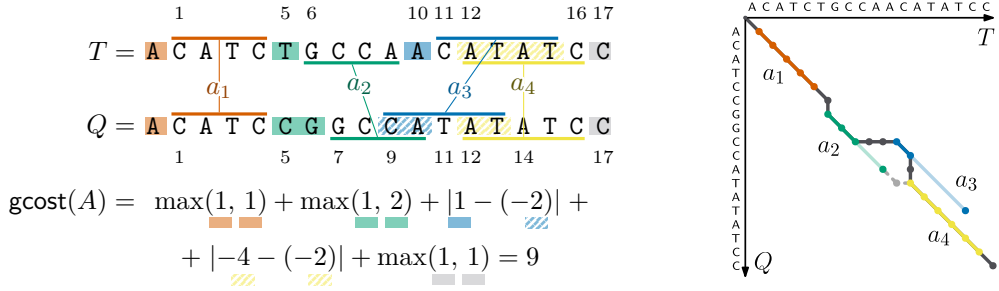}
    \end{minipage}
    \quad
    \quad
    \caption{Example of a chain $A = a_1, a_2, a_3,
      a_4$ of cost $\gcost(A) =
      9$ (see \Cref{def:connect,def:cost}), where on the left the gaps and overlaps between consecutive anchors are marked
      with solid and striped background, respectively. On the right, a possible
      alignment in the corresponding edit graph is shown that uses a prefix of
      each anchor. Skipping $a_3$
      would reduce the cost from $9$ to $6$.}\label{fig:chain}
\end{figure}

\begin{problem}[Colinear chaining~\cite{JainGT22}]\label{prob:CLC}
  Given a set $\mathcal{A}$ of $n$ exact match anchors,
  the \emph{colinear chaining problem with $\linfty$ gap costs and $\Delta_{\diag}$ overlap costs} consists of finding an ordered subset $A = a_1$, \dots, $a_c$ of $\mathcal{A}$ that is a $\prec$-chain and has minimum global cost $\gcost(A)$.
    The \emph{semi-global} variant of the problem minimizes $\sgcost(A)$ instead.
\end{problem}

\subsection{Variants, algorithms, and anchored edit distance}
\enlargethispage{2\baselineskip}
Jain, Gibney, and Thankachan~\cite{JainGT22} introduced and studied efficient solutions to \Cref{prob:CLC}, as well as variants using different precedence formulations, which we will briefly discuss now.
While the definition of non-overlapping anchors has been consistent in the literature, overlaps can be handled in multiple ways:
for example, $([1..3), [1..3)), ([2..3), [5..6))$ is a $\prec$-chain but it aligns $T[2..3)$ to two different regions of $Q$, raising doubts about the usefulness of allowing this case.
Luckily, the colinear chaining problems under $\prec$ and under the following formulations are equivalent,\footnote{We choose then to concentrate on $\prec$ and not one of its variants because the algorithmic results will be easier to prove under this precedence (see \Cref{lem:overlapcase1a} and Baker and Giancarlo~\cite{BakerG02}).} and in turn they all compute the anchored edit distance.

\begin{definition}[Alternative precedence formulations~\cite{JainGT22,RizzoCM25}]
    \label{def:alternativeprecedence}
    Anchor $a = ([\tbeg..\tend), [\qbeg..\qend))$ \emph{strictly precedes} $a' = ([\tbeg'..\tend'), [\qbeg'..\qend'))$, in symbols $a \sprec a'$, if $a\neq a'$ and $\tbeg \le \tbeg'$, $\tend \le \tend'$, $\qbeg \le \qbeg'$, $\qend \le \qend'$.
    Also, we say that $a$ \emph{strongly precedes} $a'$, in symbols $a \chainxprec a'$, if $\tbeg < \tbeg'$, $\tend < \tend'$, $\qbeg < \qbeg'$, and $\qend < \qend'$.
\end{definition}

\begin{theorem}[{\cite[Theorem 2]{JainGT22} and \cite[Theorem 3]{RizzoCM25arxiv}}]
    \label{theo:anchorededitdistance}
    For a fixed set of anchors $\mathcal{A}$, the following quantities are equal: the anchored edit distance (\Cref{def:aed}) and the optimal $\prec$-chain, $\sprec$-chain, and $\chainxprec$-chain global cost ($\gcost$, \Cref{def:cost}).
    Similarly, the semi-global variant of anchored edit distance (i.e.\ initial and final $T$-gaps cost 0) is equal to the minimum semi-global cost $\sgcost(A)$ of an optimal $\prec$-, $\sprec$-, or $\chainxprec$-chain $A$.
\end{theorem}

There are already several dynamic-programming solutions to \Cref{prob:CLC} and its variants.
Note how each precedence formulation (\Cref{def:precedence,def:alternativeprecedence}) defines a partial order on the input anchor set $\mathcal{A}$: consistently with the literature~\cite{wilbur1984context}, the chaining algorithms consider $\mathcal{A}$ to be in a topological order.
All solutions, including ours, compute values $C[j]$, the optimal partial costs
$\pgcost(A) = \sum_{i=0}^{c-1} \connect(a_i, a_{i+1})$ over all chains $A =
(a_0, a_1, \dots, a_c)$
starting in $a_0 = \startanchor$ and ending in anchor $j$ (and thus not yet connecting to the final anchor $\finalanchor$).
Due to the additivity of the $\gcost$ function, values $C[j]$ can be computed in $O(n^2)$ time via dynamic programming (DP) as $C[0]=0$ and
    \begin{align}
        C[j] &= \min_{i \in [0..j) : a_i \prec a_j} C[i] + \connect(a_i,a_j) &\forall j \in [1..n+1],\label{eq:DP}
    \end{align}
and it is easy to see that $C[n+1]$ is the minimum $\gcost$ of a $\prec$-chain.
By changing how 
$a_0$ connects to all other anchors, a variant of \Cref{eq:DP} also solves semi-global chaining.

Jain et al.\ also introduced sub-quadratic time solutions and a practical heuristic.
\begin{theorem}[{\cite[Theorem 1]{JainGT22}}]
    \label{theo:chaining}
    The colinear chaining problem with $\linfty$ gap costs and $\ddiag$ overlap costs (\Cref{prob:CLC}) can be solved in time:
    \begin{itemize}
    \item $O(n \log^2 n)$ for $\sprec$-chains with fixed-length anchors;
    \item $O(n \log^3 n)$ for $\prec$-chains; and
    \item $O(n \log^4 n)$ for $\sprec$-chains (which are now known to be equivalent to $\prec$-chains).
    \end{itemize}
\end{theorem}

Values $C[i]$ can be computed using dynamic multi-dimensional range minimum query data structures (a case of orthogonal range search in computational geometry~\cite{BergCKO08}), but this was not implemented.
Instead, \cite{JainGT22} implements a practical algorithm that \emph{approximates} the optimum $\chainxprec$-chain cost in $O(\mathrm{SOL} \cdot n + n \log n)$ expected time, where $\mathrm{SOL}$ is the cost of the output chain and we assume a uniform distribution of start positions of the anchors in $Q$ or $T$~\cite[Algorithm 2]{JainGT22}.
This solution, called \texttt{ChainX}, was tested on sequences of bacterial-like size and
features: experimentally, \texttt{ChainX}'s output on maximal unique match (MUM) anchors correlates
with edit distance and is faster than programs directly computing the edit
distance, but the performance degrades significantly when the number of anchors
increases~\cite[Section 6]{JainGT22}.
The approximate method \texttt{ChainX} was completed by Rizzo et al.\ into an exact algorithm \texttt{ChainX-opt}~\cite[Algorithm 1]{RizzoCM25}, shown to find a better score for a small number of cases with minimal practical slowdown compared to \texttt{ChainX}.
\begin{theorem}[{\cite[Theorem 2]{RizzoCM25} and \cite[Corollary
        2]{RizzoCM25arxiv}}]
    \label{theo:chainx-opt}
    The colinear chaining problems with $\linfty$ gap costs and $\ddiag$ overlap costs under $\prec$, $\wprec$, and $\chainxprec$ precedence can be solved in time $O(n \log n + \mathrm{OPT} \cdot n)$ on average, where $\mathrm{OPT}$ is the optimum chain cost, assuming that $n \le \min \lbrace \lvert Q \rvert, \lvert T \rvert \rbrace$ and the anchors are uniformly distributed in $Q$.
\end{theorem}

\parag{Properties of anchored edit distance.}
We now list a few additional properties of the anchored edit distance.
As implied by the name, the anchored edit distance is a variant on the classical
edit distance, and admits a corresponding \emph{edit graph} \cite{myers86} that contains edges
corresponding to insertions, deletions, and substitutions (all cost 1) and
matches as indicated by the anchors (of cost 0), see \cref{fig:chain}.
We define $d_{ij}$ as the distance (length of the shortest path) from $(0,0)$ to $(i, j)$.
The proof of \Cref{theo:anchorededitdistance} gives the following fact.
\begin{corollary}\label{cor:pgcost}\label{cor:pgcostdij}
    Let $\mathcal{A} = a_1, \dots, a_n$ be sorted according to any $\prec$-topological order.
    The optimal partial cost $C[j]$ to the start of an anchor $a_j$ then equals $d_{\anchorstart(a_j)}$, the smallest distance from $(0,0)$ to $\anchorstart(a_j)$ in the edit graph defined by $\mathcal{A}$.
\end{corollary}

From this correspondence, it is clear \cite[Corollary 1]{RizzoCM25}
that overlapping or touching anchors on the same diagonal can be merged, and
similarly, that an anchor of length $\ell$ can be split into $\ell$ anchors of length
$1$, or into two anchors of lengths $\ell_1 + \ell_2 = \ell$, while keeping all partial costs the same.

\subsection{Dynamic predecessor data structures}\label{subsec:predecessor}

A data structure for dynamic predecessor queries supports the following
operations on a set~$S$ of totally ordered keys:
\begin{itemize}
  \item $\mathsf{insert(x)}$: Add $x$ to $S$.
  \item $\mathsf{remove(x)}$: Remove $x$ from $S$.
  \item $\mathsf{succ(x)}$: Return $\min\{y \in S \mid y> x\}$, the smallest
    element larger than $x$.
  \item $\mathsf{pred(x)}$: Return $\max\{y \in S \mid y< x\}$, the largest
    element less than $x$.
\end{itemize}
In practice, keys can also have associated values, in which case
$\mathsf{get(x)}$ can return this value in case $x\in S$.

Common implementations use balanced binary trees (red-black trees \cite{red-black-tree})
or B-trees \cite{b-tree}, taking $O(n)$ words of space and $O(\log n)$ time for a set
of $n$ elements.
When the keys are integers up to $U$, Van Emde Boas trees
\cite{van-emde-boas-trees-2} use $O(U)$ space and allow all operations
in $O(\log \log U)$ time (in the word RAM model).
Y-fast tries \cite{xy-fast-trie} improve this to $O(\log \log U)$ time in $O(n)$ space.
Johnson's priority queue \cite{johnson-priority-queue} is another variant often
used in chaining methods that uses $O(U^{1/p})$ space for $O(p \log \log D)$
queries, where $D$ is the distance between the two elements surrounding the
queried element $x$ and $p$ is a small positive integer parameter that can be chosen freely. 

\parag{With known keys.}
In case the set of all $O(n)$ inserted keys $x$ is known in advance,
they can be sorted in $O(n \log \log n)$ time and $O(n)$ space~\cite{DBLP:conf/stoc/Han02}.
Then, each key can be assigned its \emph{rank} from $1$ to $n$ in the sorted list
(e.g.\ via a lookup table), and the predecessor structure can use the ranks instead.
Thus, this case allows for $O(\log \log n)$ rather than $O(\log \log U)$ queries
while using $O(n)$ space.

\parag{Incremental suffix minimum queries.}
We can also use predecessor structures for \emph{incremental suffix minimum
queries}, where we insert keys $x$ with an associated value $f(x)$ (but never
remove them) and get queries that ask for $\mathsf{smq}(x)=\min\{f(y) \mid y\in S : y\geq x\}$.
A standard way to support this is the following:
if we try to insert $(x, f(x))$ but $f(x)\geq \mathsf{smq}(x)$, do nothing;
otherwise, insert $(x, f(x))$ (or update the value associated with $x$), and
repeatedly find $y=\mathsf{pred}(x)$, then remove $y$ as long as $f(x) \leq f(y)$.
If $S$ is known in advance, insertions and $\mathsf{smq}$ queries can both be supported in $O(\log \log n)$ amortized time.

\section{Chaining in \texorpdfstring{$\bm{O(n \log \log n)}$}{O(n log log n)} time}
\label{sec:nlogn}
\enlargethispage{2\baselineskip}
We now give an algorithm for colinear chaining with $\linfty$ gap costs and $\Delta_{\diag}$ overlap costs on the input set $\mathcal{A} = \lbrace a_1, \dots, a_n \rbrace$ (\Cref{prob:CLC}) in $O(n \log \log n)$ time and $O(n)$ space.
We assume that $\mathcal{A}$ is sorted in lexicographical order of 2D points $\anchorstart(a)$.
We do not worry about ties, because we also assume that $\mathcal A$ does not contain perfectly-overlapping or touching anchors (on the same diagonal), as these can be merged into longer anchors without changing the anchor-supported character matches and thus the optimal $\gcost$ chain (\Cref{theo:anchorededitdistance}).
These two assumptions can be enforced in $O(n \log \log n)$ time and $O(n)$
space by sorting the anchors first by diagonal and
then by anti-diagonal using~\cite{DBLP:conf/stoc/Han02} and then merging adjacent overlapping anchors.

We follow the same general strategy as all sub-quadratic-time chaining algorithms (see e.g.\ \cite{EppsteinGGI92one,BakerG02,JainGT22}):
we partition the minimum of \Cref{eq:DP} into four disjoint cases described by the $\connect$ function and use a line sweep of the 2D plane $[0..\lvert T \rvert) \times [0..\lvert Q \rvert)$.
Consider \Cref{fig:cases} and the following equation:

\begin{align}
    C[j] &= \min \big(C^\mathrm{1a}[j], C^\mathrm{1b}[j], C^\mathrm{2}[j], C^\mathrm{3}[j] \big), \qquad\text{where} &\label{eq:fourcases}\\[2ex]
    C^{\mathrm{1a}}[j] &:= \min_{\substack{i \in [0..j) \,: \\
        a_i \prec a_j, \; \ov_T(a_i,a_j) > 0, \\ \diag(a_i) \ge \diag(a_j)}} C[i] + \diag(a_i) - \diag(a_j) &\text{(overlap, larger $T$-overlap)} \label{eq:overlap-bigger-Tov} \\
    C^{\mathrm{1b}}[j] &:= \min_{\substack{i \in [0..j) \,: \\
        a_i \prec a_j, \; \ov_Q(a_i,a_j) > 0, \\ \diag(a_i) < \diag(a_j)}} C[i] + \diag(a_j) - \diag(a_i) &\text{(overlap, larger $Q$-overlap)} \label{eq:overlap-bigger-Qov} \\
    C^{\mathrm{2}}[j]  &:= \min_{\substack{i \in [0..j) : \\     \anchorend(a_i)\preceq \anchorstart(a_j), \\ \diag(a_i) > \diag(a_j)}} C[i] + \Delta_Q(a_i,a_j) &\text{(gap-gap, larger $Q$-gap)} \label{eq:gap-gap-bigger-Qgap}\\
    C^{\mathrm{3}}[j]  &:= \displaystyle \min_{\substack{i \in [0..j) : \\     \anchorend(a_i)\preceq \anchorstart(a_j), \\ \diag(a_i) \le \diag(a_j)}} C[i] + \Delta_T(a_i,a_j) &\text{(gap-gap, larger $T$-gap)} \label{eq:gap-gap-bigger-Tgap}
\end{align}
where $\min(\emptyset) = +\infty$. 
Then, our algorithm can be divided into three main stages, as partially described in \Cref{algo:main}:

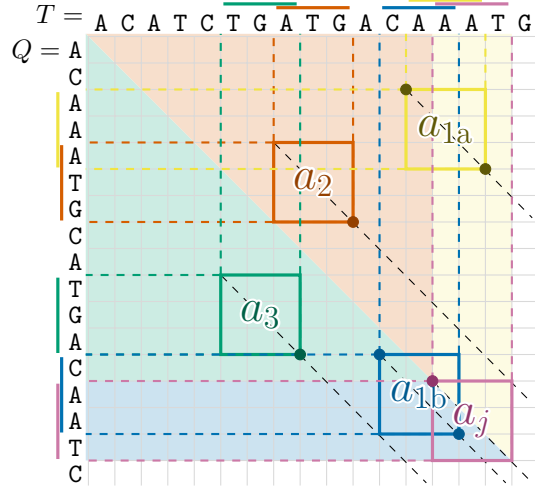
\begin{figure}[t]
    \centering
    \begin{tikzpicture}
	\tikzset{
		invisible/.style = {minimum size=0pt, inner sep=0pt}
	}

\matrix[
	matrix of nodes,
	minimum size=2mm,
	column sep={3.5mm,between origins},
	row sep={3.5mm,between origins},
	inner sep=1pt,
	nodes={anchor=base},
	anchor=south west,
	font=\ttfamily,
] (Q) at (0,0) {
A&C&A&T&C&T&G&A&T&G&A&C&A&A&A&T&G\\
};

\matrix[
	matrix of nodes,
	minimum size=2mm,
	column sep={3.5mm,between origins},
	row sep={3.5mm,between origins},
	inner sep=1pt,
	nodes={anchor=base},
	anchor=north east,
	font=\ttfamily,
] (T) at (0,0) {
A\\C\\A\\A\\A\\T\\G\\C\\A\\T\\G\\A\\C\\A\\A\\T\\C\\
};

\draw[-,very thick, colorA] ($(Q-1-8.north west) + (0,.05cm)$) -- node[invisible] (Qanchor0) {} ($(Q-1-10.north east) + (0,.05cm)$);
\draw[-,very thick, colorA] ($(T-5-1.north west) + (-.05cm,0)$) -- node[invisible] (Tanchor0) {} ($(T-7-1.south west) + (-.05cm,0)$);

\node[invisible] at ($(Q-1-8.west|-T-5-1.north) + (-1pt,1pt)$) (a0northwest) {};
\node[invisible] at ($(Q-1-10.west|-T-7-1.north) + (3.5mm,-3.5mm) + (-1pt,1pt)$) (a0southeast) {};
\draw[dashed] (a0northwest) -- ($(Q-1-17.east|-T-14-1.south) + (-1pt,1pt)$);
\draw[-,very thick,draw=colorA] (a0northwest) rectangle (a0southeast) node[invisible,pos=.5] (a0) {};

\draw[dashed,thick,-,colorA] (a0northwest) -- ($(Q-1-8.south west-|a0northwest)$);
\draw[dashed,thick,-,colorA] (a0southeast) -- ($(Q-1-10.south east-|a0southeast)$);
\draw[dashed,thick,-,colorA] (a0northwest) -- ($(T-5-1.north east|-a0northwest)$);
\draw[dashed,thick,-,colorA] (a0southeast) -- ($(T-7-1.south east|-a0southeast)$);

\draw[-,very thick, colorB] ($(Q-1-6.north west) + (0,.10cm)$) -- node[invisible] (Qanchor1) {} ($(Q-1-8.north east) + (0,.10cm)$);
\draw[-,very thick, colorB] ($(T-10-1.north west) + (-.10cm,0)$) -- node[invisible] (Tanchor1) {} ($(T-12-1.south west) + (-.10cm,0)$);

\node[invisible] at ($(Q-1-6.west|-T-10-1.north) + (-1pt,1pt)$) (a1northwest) {};
\node[invisible] at ($(Q-1-8.west|-T-12-1.north) + (3.5mm,-3.5mm) + (-1pt,1pt)$) (a1southeast) {};
\draw[dashed] (a1northwest) -- ($(Q-1-13.east|-T-17-1.south) + (-1pt,1pt)$) node[pos=1,invisible] (diagonalbottom1) {};
\draw[-,very thick,draw=colorB] (a1northwest) rectangle (a1southeast) node[invisible,pos=.5] (a1) {};

\draw[dashed,thick,-,colorB] (a1northwest) -- ($(Q-1-6.south west-|a1northwest)$);
\draw[dashed,thick,-,colorB] (a1southeast) -- ($(Q-1-8.south east-|a1southeast)$);
\draw[dashed,thick,-,colorB] (a1northwest) -- ($(T-10-1.north east|-a1northwest)$);
\draw[dashed,thick,-,colorB] (a1southeast) -- ($(T-12-1.south east|-a1southeast)$);

\draw[-,very thick, colorC] ($(Q-1-12.north west) + (0,.05cm)$) -- node[invisible] (Qanchor2) {} ($(Q-1-14.north east) + (0,.05cm)$);
\draw[-,very thick, colorC] ($(T-13-1.north west) + (-.05cm,0)$) -- node[invisible] (Tanchor2) {} ($(T-15-1.south west) + (-.05cm,0)$);

\node[invisible] at ($(Q-1-12.west|-T-13-1.north) + (-1pt,1pt)$) (a2northwest) {};
\node[invisible] at ($(Q-1-14.west|-T-15-1.north) + (3.5mm,-3.5mm) + (-1pt,1pt)$) (a2southeast) {};
\draw[dashed] (a2northwest) -- ($(Q-1-16.east|-T-17-1.south) + (-1pt,1pt)$) node[pos=1,invisible] (diagonalbottom2) {};

\draw[-,very thick,draw=colorC] (a2northwest) rectangle (a2southeast) node[invisible,pos=.5] (a2) {};

\draw[dashed,thick,-,colorC] (a2northwest) -- ($(Q-1-12.south west-|a2northwest)$);
\draw[dashed,thick,-,colorC] (a2southeast) -- ($(Q-1-14.south east-|a2southeast)$);
\draw[dashed,thick,-,colorC] (a2northwest) -- ($(T-13-1.north east|-a2northwest)$);
\draw[dashed,thick,-,colorC] (a2southeast) -- ($(T-15-1.south east|-a2southeast)$);

\draw[-,very thick, colorD] ($(Q-1-13.north west) + (0,.15cm)$) -- node[invisible] (Qanchor3) {} ($(Q-1-15.north east) + (0,.15cm)$);
\draw[-,very thick, colorD] ($(T-3-1.north west) + (-.10cm,0)$) -- node[invisible] (Tanchor3) {} ($(T-5-1.south west) + (-.10cm,0)$);

\node[invisible] at ($(Q-1-13.west|-T-3-1.north) + (-1pt,1pt)$) (a3northwest) {};
\node[invisible] at ($(Q-1-15.west|-T-5-1.north) + (3.5mm,-3.5mm) + (-1pt,1pt)$) (a3southeast) {};
\draw[dashed] (a3northwest) -- ($(Q-1-17.east|-T-7-1.south) + (-1pt,1pt)$);
\draw[-,very thick,draw=colorD] (a3northwest) rectangle (a3southeast) node[invisible,pos=.5] (a3) {};

\draw[dashed,thick,-,colorD] (a3northwest) -- ($(Q-1-13.south west-|a3northwest)$);
\draw[dashed,thick,-,colorD] (a3southeast) -- ($(Q-1-15.south east-|a3southeast)$);
\draw[dashed,thick,-,colorD] (a3northwest) -- ($(T-3-1.north east|-a3northwest)$);
\draw[dashed,thick,-,colorD] (a3southeast) -- ($(T-5-1.south east|-a3southeast)$);




\draw[-,very thick, colorG] ($(Q-1-14.north west) + (0,.10cm)$) -- node[invisible] (Qanchor5) {} ($(Q-1-16.north east) + (0,.10cm)$);
\draw[-,very thick, colorG] ($(T-14-1.north west) + (-.10cm,0)$) -- node[invisible] (Tanchor5) {} ($(T-16-1.south west) + (-.10cm,0)$);

\node[invisible] at ($(Q-1-14.west|-T-14-1.north) + (-1pt,1pt)$) (a5northwest) {};
\node[invisible] at ($(Q-1-16.west|-T-16-1.north) + (3.5mm,-3.5mm) + (-1pt,1pt)$) (a5southeast) {};
\draw[dashed] (a5northwest) -- ($(Q-1-17.east|-T-17-1.south) + (-1pt,1pt)$);
\draw[-,very thick,draw=colorG] (a5northwest) rectangle (a5southeast) node[invisible,pos=.5] (a5) {};

\draw[dashed,thick,-,colorG] (a5northwest) -- ($(Q-1-14.south west-|a5northwest)$);
\draw[dashed,thick,-,colorG] (a5southeast) -- ($(Q-1-16.south east-|a5southeast)$);
\draw[dashed,thick,-,colorG] (a5northwest) -- ($(T-14-1.north east|-a5northwest)$);
\draw[dashed,thick,-,colorG] (a5southeast) -- ($(T-16-1.south east|-a5southeast)$);

\node[colorAdark] at (a0) {\LARGE\backgroundcontour{$a_2$}};
\node[colorAdark] at (a0) {\LARGE$a_2$};

\node[colorBdark] at (a1) {\LARGE\backgroundcontour{$a_3$}};
\node[colorBdark] at (a1) {\LARGE$a_3$};

\node[colorCdark] at (a2) {\LARGE\backgroundcontour{$a_{\mathrm{1b}}$}};
\node[colorCdark] at (a2) {\LARGE$a_{\mathrm{1b}}$};

\node[colorDdark] at (a3) {\LARGE\backgroundcontour{$a_{\mathrm{1a}}$}};
\node[colorDdark] at (a3) {\LARGE$a_{\mathrm{1a}}$};


\node[colorGdark] at (a5) {\LARGE\backgroundcontour{$a_j$}};
\node[colorGdark] at (a5) {\LARGE$a_j$};

\node[anchor=base east,inner sep=0pt] at (Q.west) {$T =\,$};
\node[anchor=base east,inner sep=0pt] at (T-1-1.base west) {$Q =\,$};
\node at (Q.north east) {};
\node at (T.south west) {};

\begin{scope}[inner sep=0pt,circle,minimum size=4pt]
\node[fill=colorAdark,draw=colorAdark] at (a0southeast) {};
\node[fill=colorBdark,draw=colorBdark] at (a1southeast) {};
\node[fill=colorCdark,draw=colorCdark] at (a2northwest) {};
\node[fill=colorCdark,draw=colorCdark] at (a2southeast) {};
\node[fill=colorDdark,draw=colorDdark] at (a3northwest) {};
\node[fill=colorDdark,draw=colorDdark] at (a3southeast) {};
\node[fill=colorGdark,draw=colorGdark] at (a5northwest) {};
\end{scope}

\begin{scope}[on background layer]
    \fill[colorA!20] (a5northwest) -- (Q-1-1.south west-|T-1-1.north east) -- (a5northwest|-Q-1-1.south west) -- cycle;
    \fill[colorB!20] (a5northwest) -- (Q-1-1.south west-|T-1-1.north east) -- (a5northwest-|T-1-1.north east) -- cycle;
    \fill[colorC!20] (a5northwest) -- (T-14-1.north east|-a5northwest) -- (T-16-1.north east|-a5southeast) -- (a5southeast) -- (a5northwest);
    \fill[colorD!15] (a5northwest) -- (Q-1-14.south west-|a5northwest) -- (a5southeast|-Q-1-16.south east) -- (a5southeast) -- (a5northwest);
\end{scope}

\begin{scope}[on background layer]
\draw[black!15,step=3.5mm] (0,0) grid ($(Q-1-17.south east|-T-17-1.south east) + (1.1pt,-.01pt)$);
\end{scope}

\end{tikzpicture}
    \caption{Anchors $a_{\mathrm{1a}}$, $a_{\mathrm{1b}}$ $a_{\mathrm{2}}$, $a_{\mathrm{3}}$ showing the four mutually exclusive cases for the computation of $C[j]$ (\Cref{eq:fourcases}): overlap case with bigger $T$-overlap (1a) or bigger $Q$-overlap (1b), gap-gap with bigger $Q$-gap (2), and gap-gap with bigger $T$-gap (3).
    The areas containing the endpoint for each recursive anchor $a_i$ are highlighted with different colors, and the startpoints of $a_{\mathrm{1a}}$, $a_{\mathrm{1b}}$ (\tikz{\node[circle,inner sep=0pt,fill=colorDdark,minimum size=4pt] {};} and \tikz{\node[circle,inner sep=0pt,fill=colorCdark,minimum size=4pt] {};} markers) are also highlighted to indicate that in \Cref{eq:overlap-bigger-Qov,eq:overlap-bigger-Tov} they need to be checked.}\label{fig:cases}
\end{figure}

\begin{description}
\item[Pre-processing.]

  First, we compute the sorted orders for the ``left-to-right'' (i.e.\ $\tbeg$ and $\tend$) and ``top-to-bottom'' (i.e.\ $\qbeg$ and $\qend$) line sweeps of the set of anchor startpoints and endpoints, and we additionally sort the anchors by their diagonal value $\diag(a)$ to compute $\rank_{\diag(a)}$, where $\rank_{\diag(a)} = d$ if $\diag(a)$ is the $d$-th hit diagonal.
    As explained in \Cref{subsec:predecessor}, this can be done in $O(n \log \log n)$ time and $O(n)$ space, since sorting considers at most $2n$ points.
    Following Baker and Giancarlo~\cite{BakerG02}, for each $j \in [1..n]$ we compute anchors $\hat{a}_j$ and $\overleftarrow{a}_{\!\!j}$, respectively the closest anchor overlapping $\anchorstart(a_j)$ in a bigger diagonal (i.e.\ ``above'') and a smaller diagonal (i.e.\ ``left''), with two line sweeps as detailed in \Cref{subsec:overlaps} ($\hat{a}_j$ and $\overleftarrow{a}_{\!\!j}$ are called $\mathit{visl}(a_j)$ and $\mathit{visa}(a_j)$ in \cite{BakerG02}).
    This will allow handling both overlap cases in constant time during the main line sweep.

    \item[Main line sweep.] We compute values $C[j]$ while performing a line sweep of the anchor startpoints and endpoints, from smaller to larger $T$-coordinate value.
    Following Eppstein et al.~\cite{EppsteinGGI92one} (see also \cite{DBLP:journals/algorithmica/ChaoM95,BakerG02}), when a startpoint $\anchorstart(a_j)$ is processed, $C[j]$ is calculated as per \Cref{eq:fourcases} using $\hat{a}_j$, $\overleftarrow{a}_{\!\!j}$, and the data structures for the two gap-gap cases 2 and 3 as we will describe in \Cref{sub:biggerQgap,sub:biggerTgap}.
    When a startpoint or endpoint is reached (endpoints are processed before startpoints if they fall on the same $T$-position), these gap-gap data structures are updated. 
    \item[Post-processing.] Value $C[n+1]$ is equal to the anchored edit distance. Instead of augmenting our solution to store backtracking information (see e.g.\ \cite{RizzoCM25}), we can optionally recover an optimal $\prec$-chain from values $C[j]$ in $O(n)$ time: we start from $j = n+1$ and we iteratively perform a linear search to find a predecessor $i < j$ such that $C[i] + \connect(a_i,a_j) = C[j]$ and we set $j = i$, until $j = 0$.  
\end{description}

\begin{algorithm}[t]
    \SetKwInOut{Input}{Input}
    \SetKwInOut{Output}{Output}
    \Input{Anchors $\mathcal{A} = \lbrace a_1, \dots, a_n \rbrace$ between $T$ and $Q$ such that no two anchors on the same diagonal overlap or have an empty gap (i.e.\ $\Delta_T > 0$ or $\Delta_Q > 0$)}
    \Output{The minimum global chaining cost ($\gcost$) over all $\prec$-chains}
    Sort anchors $a \in \mathcal{A}$ by diagonal value $\diag(a)$, to compute $\mathcal{D} := \lbrace \diag(a) : a \in \mathcal{A} \rbrace$\;
    $\text{\texttt{preprocess\_1a}}(\mathcal{A})$ \Comment*[r]{Computes anchors $\hat{a}$}
    $\text{\texttt{preprocess\_1b}}(\mathcal{A})$\Comment*[r]{Computes anchors $\overleftarrow{a}$}
    Sort startpoints and endpoints $P = \lbrace (a_j.\mathsf{start},j) : a_j \in \mathcal{A} \rbrace \cup \lbrace (a_j.\mathsf{end},j) : a_j \in \mathcal{A} \rbrace \cup \lbrace ((0,0),0), ((\lvert T \rvert, \lvert Q \rvert), n+1) \rbrace$ in increasing $T$-coordinate order (in breaking ties, endpoints are smaller than startpoints and $\prec$ is used for points of the same type)\;
    $\text{\texttt{init\_case\_2}}(\mathcal{A})$\;
    $\text{\texttt{init\_case\_3}}(\mathcal{A})$\;
    $C[0] \gets 0$\;
    \For{$(p,j) \in P$ in sorted order}{%
        \uIf{$p = a_j.\mathsf{start}$}{%
            $C^\mathrm{1a}[j] \gets C[\hat{a}_j] + \ddiag(\hat{a}_j, a_j)$\Comment*[r]{\Cref{lem:overlapcase1a}}
            $C^\mathrm{1b}[j] \gets C[\overleftarrow{a}_{\!\!j}] + \ddiag(\overleftarrow{a}_{\!\!j}, a_j)$\Comment*[r]{\Cref{lem:overlapcase1a}}
            $C^2[j] \gets \text{\texttt{process\_startpoint\_case\_2}}(a_j)$\;
            $C^3[j] \gets \text{\texttt{process\_startpoint\_case\_3}}(a_j)$\;
            $C[j] = \min \big(C^\mathrm{1a}[j], C^\mathrm{1b}[j], C^\mathrm{2}[j], C^\mathrm{3}[j] \big)$\Comment*[r]{\Cref{eq:fourcases}}
        }
        \Else(\Comment*[f]{$p = a_j.\mathsf{end}$}){%
            $\text{\texttt{process\_endpoint\_case\_2}}(a_j)$\;
            $\text{\texttt{process\_endpoint\_case\_3}}(a_j)$\;
        }
    }
    \Return $C[n+1]$\Comment*[r]{Best $\pgcost$ to final dummy anchor $a_{n+1} = \finalanchor$}
    \caption{High-level description of chaining with $\linfty$ gap costs and $\ddiag$ overlap costs.}\label{algo:main}
\end{algorithm}

The rest of this section is dedicated to showing how to solve the recursive cases individually, in order of hardness: the overlap cases (1a and 1b), the bigger $Q$-gap case (2), and the bigger $T$-gap case (3).
The correctness will follow from \Cref{eq:fourcases} by induction on $j = 0, \dots, n+1$, obtaining our main theoretical result thanks to \Cref{theo:anchorededitdistance}.
\begin{theorem}\label{theo:nloglogn}
    The anchored edit distance (\Cref{def:aed}) can be computed in $O(n \log \log n)$ time and $O(n)$ space, by solving the colinear chaining problem with $\linfty$ gap costs and $\ddiag$ overlap costs (\Cref{prob:CLC}) in the same time and space.
\end{theorem}

\subsection{Overlap cases 1a and 1b}\label{subsec:overlaps}
We study the computation of values $C^{\mathrm{1a}}[j]$ (\Cref{eq:overlap-bigger-Tov}), as case $C^{\mathrm{1b}}$ is solved in a symmetrical way.
We are in this case when $a_i \prec a_j$, $\ov_T(a_i, a_j) > 0$, and $\diag(a_i) \ge \diag(a_j)$, meaning that the anchors from $C^{\mathrm{1a}}[j]$ are exactly anchors $a_i$ such that: ($i$) their $T$-interval $[\tbeg^i..\tend^i)$ contains $\tbeg^j$; and ($ii$) $\diag(a_i) \ge \diag(a_j)$.
Note that conditions $i$ and $ii$ are sufficient for the colinearity $a_i \prec a_j$.
By adapting an argument of Baker and Giancarlo~\cite[Section 3.1]{BakerG02}, we can prove that we only need to check the closest anchor above $a_j$, as all other anchors cannot be a better case-1a predecessor.

\begin{lemma}\label{lem:overlapcase1a}
    In the overlap case with bigger $T$-overlap (\Cref{eq:overlap-bigger-Tov}) and under the assumption that no two anchors overlap on the same diagonal, the optimal partial cost $C^\mathrm{1a}[j]$ to anchor $a_j$ is equal to $C[\hat{a}] + \ddiag(\hat{a},a_j)$, where $\hat{a}$ is the anchor with the smallest $\diag(\hat{a})$ bigger than $\diag(a_j)$ such that $\tbeg^j \in [\hat{\tbeg}..\hat{\tend})$.
    If $\hat{a}$ is not defined, then $C^\mathrm{1a}[j] = +\infty$.
\end{lemma}
\begin{proof}
    Assume by contradiction that $a' = ([\tbeg'..\tend'),[\qbeg'..\qend'))$ is a strictly better overlapping anchor considered by case 1a which is farther away, that is, $\tbeg^j \in [\tbeg'..\tend')$, $\diag(a_j) \le \diag(\hat{a}) < \diag(a')$ (no two anchors overlap in the same diagonal by assumption), and $C[a'] + \diag(a') - \diag(a_j) < C[\hat{a}] + \diag(\hat{a}) - \diag(a_j)$.
    If $a' \prec \hat{a}$, we reach a contradiction, as $C[a'] + \connect(a',\hat{a}) < C[\hat{a}]$ (case 1a) violates the definition of $C[\hat{a}]$ (\Cref{eq:DP}).
    Otherwise $\tbeg' \in (\hat{\tbeg}..\hat{\tend})$ and we can split anchor $\hat{a}$ into two anchors $\hat{a}^1$ and $\hat{a}^2$, the first covering $[\hat{\tbeg}..\tbeg')$ and the second covering $[\tbeg'..\hat{\tend})$.
    Due to \Cref{cor:pgcostdij}, the partial optimal costs $C'$ of the modified instance $(\mathcal{A} \setminus \lbrace \hat{a} \rbrace) \cup \lbrace \hat{a}^1, \hat{a}^2 \rbrace$ are identical for all anchors but $C'[\hat{a}^1] = C'[\hat{a}^2] = C[\hat{a}]$.
    But then $a' \prec \hat{a}^2$ by construction and $C'[a'] + \connect(a',\hat{a}^2) < C'[\hat{a}^2]$, reaching a contradiction.
\end{proof}

Then, as in \cite[Section 3.3]{BakerG02}, we can compute anchors $\hat{a}_j$ with a $T$-coordinate sweep line and a successor data structure on the set $S \subseteq \mathcal{D} := \lbrace \diag(a) : a \in \mathcal{A} \rbrace$ of active anchors overlapping with the sweep line, as implemented in \Cref{preprocess:case1a}.
Since $\lvert \mathcal{D} \rvert \le n$ and $\mathcal{D}$ can be precomputed, the computation takes $O(n \log \log n)$ time and $O(n)$ space.
The overlap case 1b is symmetrical and can be obtained by sorting the points by $Q$-coordinate and using predecessor queries (see line 8 of \Cref{preprocess:case1a}).

\begin{algorithm}[t]
    \SetKwProg{myproc}{Procedure}{}{}
    \myproc{$\textnormal{\texttt{preprocess\_1a}}(\mathcal{A})$}{
        Sort startpoints and endpoints $P = \lbrace (a.\mathsf{start},a) : a \in \mathcal{A} \rbrace \cup \lbrace (a.\mathsf{end},a) : a \in \mathcal{A} \rbrace$ in increasing $T$-coordinate order (in breaking ties, endpoints are smaller)\;
        $\hat{a}_0 \gets \perp$\;
        $\hat{a}_{n+1} \gets \perp$\;
        $S \gets \emptyset$\Comment*[r]{predecessor data structure over $\mathcal{D} = \lbrace \diag(a) : a \in \mathcal{A} \rbrace$}
        \For{$(p,a) \in P$ in sorted order}{%
            \uIf{$p = a.\mathsf{start}$}{%
                    $\hat{a} \gets S.\mathsf{succ}(\diag(a)).\mathsf{value}$ \Comment*[r]{if it exists, otherwise $\perp$}
                $S.\mathsf{insert}(\diag(a) \mapsto a)$\;
            }
            \Else(\Comment*[f]{$p = a.\mathsf{end}$}){%
                $S.\mathsf{remove}(\diag(a))$\;
            }
        }
    }
    \caption{Procedure handling the computation of anchors $\hat{a}_j$ for $j \in [1..n]$, where $\hat{a}_j$ is defined as the anchor overlapping with $a_j$ with the smallest diagonal greater than $\diag(a_j)$, if it exists, otherwise $\hat{a} = \perp$, the dummy anchor. Note that we assume anchor set $\mathcal{A}$ to not have any anchors overlapping or touching along the same diagonal, so no two startpoints or endpoints can be equal.}\label{preprocess:case1a}
\end{algorithm}

\subsection{Gap-gap case, bigger gap in \texorpdfstring{$\bm Q$}{Q}}\label{sub:biggerQgap}
If we are in this case, that is, if $\anchorend(a_i)\preceq \anchorstart(a_j)$ and $\diag(a_i) > \diag(a_j)$, then $a_i \prec a_j$ and
\begin{align}
    C^2[j] = \min_{\substack{i \in [0..j) : \\ \anchorend(a_i)\preceq\anchorstart(a_j), \\ \diag(a_i) > \diag(a_j)}} \big( C[i] + \qbeg^j - \qend^i \big) 
    = \qbeg^j +\min_{\substack{i \in [0..j) : \\ \anchorend(a_i)\preceq\anchorstart(a_j),  \\ \diag(a_i) > \diag(a_j)}} \big( C[i] - \qend^i \big).\label{eq:casetwofinal} 
\end{align}
In geometric terms, we need to consider each anchor whose endpoint is in the first octant from $\anchorstart(a_j)$, as shown in \Cref{fig:cases}.
This octant is characterized by the two equations $\tend^i \le \tbeg^j$ $(i)$ and $\diag(a_i) > \diag(a_j)$ $(ii)$: the former is guaranteed by processing the points from left to right, and the diagonal constraint is handled via a suffix minimum query over diagonal keys, as shown in \Cref{algo:casetwo}.
Indeed, from $(i)$ and $(ii)$ we can derive that $\qend^j \le \qbeg^i$ and thus $\anchorend(a_i) \preceq \anchorstart(a_j)$ as follows:
\begin{align*}
    \diag(a_i) > \diag(a_j) 
    &\,\Longleftrightarrow\,
    \tend^i - \qend^i > \tbeg^j - \qbeg^j \\
    &\,\Longleftrightarrow\,
    \qbeg^j - \qend^i > \tbeg^j - \tend^i \ge 0 &(\tend^i \le \tbeg^j) \\
    &\,\,\Longrightarrow\, \qend^j \le \qbeg^i.
\end{align*}

We briefly discuss \Cref{algo:casetwo}.
As originally presented in \cite[Sections 2 and 4]{EppsteinGGI92one} for similar recursions, the recursive values $C[i] - \qend^i$ are stored in a one-dimensional data structure, as our line sweep conveniently gives us one half-plane for free.
Since the set of $\diag(a)$ keys is known in advance and the number of startpoints and endpoints is $2n$, we can handle case 2 in $O(n \log \log n)$ time and $O(n)$ space.

\begin{algorithm}[t]
    \SetKwProg{myproc}{Procedure}{}{}
    \myproc{$\textnormal{\texttt{init\_case\_2}}(\mathcal{A})$}{
        $S \gets \emptyset$\Comment*[r]{suffix min query structure over $\mathcal{D} = \lbrace \diag(a) : a \in \mathcal{A} \rbrace$}
        $S.\mathsf{insert}(\diag(a_0) \mapsto 0)$\Comment*[r]{base case for $a_0 = \startanchor$}
    }
    \myproc{$\textnormal{\texttt{process\_startpoint\_case\_2}}(a_j)$}{
        \Return $\qbeg^j + S.\mathsf{smq}(\diag(a_j) + 1)$\;
    }
    \myproc{$\textnormal{\texttt{process\_endpoint\_case\_2}}(a_i)$}{
        $S.\mathsf{insert}(\diag(a_i) \mapsto C[i] - \qend^i)$\Comment*[r]{insert or update the recursive value}
    }
    \caption{Procedures used in \Cref{algo:main} to handle the computation of $C^2[j]$, the optimal partial cost $\pgcost(A)$ of any $\prec$-chain $A$ ending at anchor $a_j$ with a predecessor $a_i$ in the gap-gap, bigger $Q$-gap case (see \Cref{fig:cases} and \Cref{eq:fourcases}).}\label{algo:casetwo}
\end{algorithm}

\subsection{Gap-gap case, bigger gap in \texorpdfstring{$\bm T$}{T}}\label{sub:biggerTgap}
This is the final and hardest case to handle for the computation of \Cref{eq:fourcases}.
Even though it is symmetrical to case 2, we do not get a half-plane filter for
free as our sweeping of the plane goes from left to right (see
\Cref{fig:cases}) while the constraint is on $q$, and this cannot be easily handled with a one-dimensional data structure or by changing the sweep line.
Indeed, in recent chaining results this type of recursive case has been handled with 2D search data structures, resulting in final time complexities of $O(n \log^2 n)$ or $O(n \log n \log \log n)$ ~\cite{DBLP:journals/almob/OttoHGS11,DBLP:journals/ploscb/RenC21,JainGT22}.
Instead, using the space-partitioning technique of Eppstein et al.~\cite{EppsteinGGI92one}, we can efficiently represent the optimal recursive cases projected onto the sweep line as dynamic one-dimensional segments, as described in the rest of this section.
Our goal is to present a sufficiently detailed variant of the $O(n \log \log n)$-time technique that runs in $O(n)$ space, which is not a widely known result and might be of independent interest.\footnote{The original efficient chaining solutions by Eppstein et al.~\cite{EppsteinGGI92one} and Chao and Miller~\cite{DBLP:journals/algorithmica/ChaoM95} used $O(\lvert T \rvert + \lvert Q \rvert)$ space. Baker and Giancarlo do claim to use $O(n)$ space~\cite[Theorem 4.1]{BakerG02} but do not provide an explicit space-complexity proof for handling the intersection points~\cite[p.\ 251]{BakerG02} as we will later do in \Cref{lem:case3}. Other methods use $O(n)$ space but require $O(n \log \log (|T|+|Q|))$ time.}

Symmetrically to \Cref{eq:casetwofinal}, we can rewrite $C^3[j]$ as
\begin{align}
    C^3[j] &=
    \tbeg^j + \min_{\substack{i \in [0..j) : \\ \anchorend(a_i)\preceq \anchorstart(a_j), \\ \diag(a_i) \le \diag(a_j)}} \big( C[i] - \tend^i \big).\label{eq:casethreefinal}
\end{align}
From this we can derive the following first main insight: each anchor $a_i$ defines a \emph{forward region of influence} where it can precede a following anchor $a_j$ in the sense of \cref{eq:casethreefinal}.
This region (also called ``area of influence'') is triangular and delimited by the horizontal and diagonal lines starting from endpoint $\anchorend(a_i)$ (see \cref{fig:viz}): for any successor $a_j$ with $\anchorstart(a_j) = (t, q)$, we consider the intersection of the half-plane $\qend^i \le q$ bounded by the horizontal line through $\anchorend(a_i)$ and the half-plane $\diag(\tbeg^j, \qbeg^j) \geq \diag(a_i)$ bounded by the diagonal line through $\anchorend(a_i)$.
When the areas of multiple anchors intersect, those with the smaller value $C[i] - \tend^i$ take precedence, as remarked by the following observation.
\begin{observation}[{Chao and Miller~\cite[Lemma 1]{DBLP:journals/algorithmica/ChaoM95}}]\label{obs:intersectionare}
    Consider any two anchors $a_i = ([\tbeg^i..\tend^i),[\qbeg^i..\qend^i))$ and $a_{i'} = ([\tbeg^{i'}..\tend^{i'}),[\qbeg^{i'}..\qend^{i'}))$.
    For all anchors $a_j$ such that $a_j$ is case-$3$ successor of both $a_i$ and $a_{i'}$ (as in \Cref{eq:casethreefinal}), in symbols $\anchorend(a_i) \preceq \anchorstart(a_j)$,  $\anchorend(a_{i'}) \preceq \anchorstart(a_j)$, $\diag(a_i) \le \diag(a_j)$, and $\diag(a_{i'}) \le \diag(a_j)$, we have that the value $\min(C[i] - \tend^i, C[i'] - \tend^{i'})$ does not depend on the exact position of $a_j$.
\end{observation}

Let $a_j = ([\tbeg^j..\tend^j),[\qbeg^j..\qend^j))$ be an arbitrary anchor starting at the current sweep-line coordinate $t$, that is, $\tbeg^j = t$, and consider \Cref{fig:viz}.
\begin{figure}[t]
    \centering
    \input{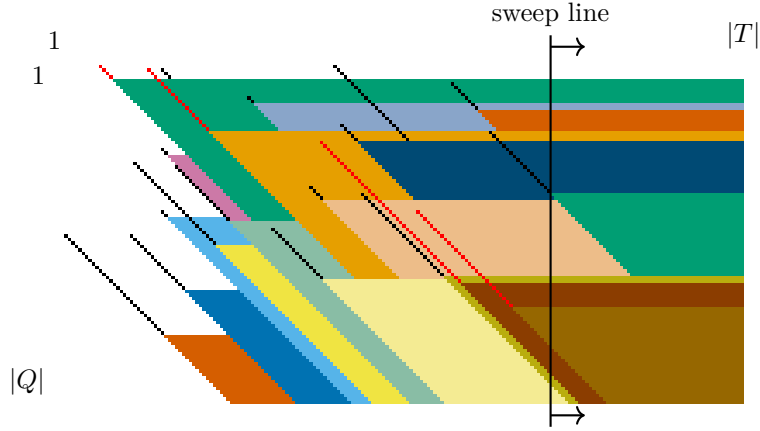}
    \caption{Optimal recursive values for case 3 (\Cref{eq:fourcases}) handled in the main left-to-right sweep line.
    The anchors of an optimal $\prec$-chain are shown in red, whereas the other anchors of $\mathcal{A}$ are black.
    Each color corresponds to a region ``owned'' by a different anchor $a_i$, so values $R_t[q]$ of the owner array (\Cref{eq:ownerarray}) can be deduced by the color at the sweep line.}
    \label{fig:viz}
\end{figure}
We now define how the optimal recursive anchors for $a_j$ change based on where the startpoint of $a_j$ is along the sweep line (i.e.\ its $Q$-coordinate $\qbeg^j$) as follows, with ties in the $\argmin$ broken to larger $i$:
\begin{equation}
    R_t[q] = \argmin_{\substack{i \in [0..n) : \\ \tend^i \le t, \, \diag(a_i) \le t - q}} \big( C[i] - \tend^i \big), \qquad\qquad\text{where}\; \argmin (\emptyset) = \perp \quad\text{and}\quad q \in [0..\lvert Q \rvert].\label{eq:ownerarray}
\end{equation}
We call $R_t$ the \emph{owner array}, because
value $R_t[q]$ is equal to the anchor index the area belongs to: by construction we have $C^3[j] = \tbeg^j + C[i] - \tend^i$ with $i = R_{\tbeg^j}[\qbeg^j]$.
Note that when there is no anchor $a_i$ such that $\anchorend(a_i) \preceq (t,q)$, then $R_t[q] = \perp$, where $\perp$ is the dummy anchor, and $C^3[j] = +\infty$.

The owner array has size $O(\lvert Q \rvert)$, so it would be too expensive to maintain it as a flat array.
Instead, Eppstein et al.~\cite{EppsteinGGI92one} employ a compact $O(n)$-space version of $R_t$ as a doubly-linked list (called OWNER list) using the following second main insight: while the areas of influence can have complex interactions and shapes as seen in \Cref{fig:viz}, they are always delimited by some horizontal or diagonal line starting in an anchor endpoint.

\enlargethispage{2em}

\begin{observation}[{\cite{EppsteinGGI92one}}]\label{obs:dividinglines}
    For all changes in anchor index $R_t[q] \neq R_t[q+1]$ with $q \in [1..\lvert Q \rvert]$, there is a unique \emph{active dividing line} 
    such that the following holds for some unique anchor $a_i$, $i \in [0..n]$:
    \begin{itemize}
    \item (horizontal dividing line) the active line is $q+1=\qend^i$ and $R_t[q+1] = i$, that is, $a_i$ owns the region of influence \emph{starting} at $q+1$; or
    \item (diagonal dividing line) the active line is $\diag(t, q) = \diag(a_i)$ and $R_t[q] = i$, that is, $a_i$ owns the area of influence \emph{ending} at $q$.
    \end{itemize}
\end{observation}

Then, as shown in \Cref{algo:casethree}, we can represent $R_t$ in a compressed $O(n)$-space representation consisting of two predecessor data structures of diagonals mapped to anchors and a doubly-linked list of active lines from \Cref{obs:dividinglines}.
These lines, which are separately indexable via two dynamic predecessor data structures (\Cref{subsec:predecessor}), are used to partition the sweep line regions with different $R_t[q]$ values.
Note that \Cref{obs:dividinglines} does not indicate which anchor owns the region between a diagonal and a horizontal line, hence we need to also store the different values of $R_t[q]$ as associated values of the elements in the list.
We obtain the following result and thus complete the proof of \Cref{theo:nloglogn}.

\begin{algorithm}[pt]
    \SetKwProg{myproc}{Procedure}{}{}
    \myproc{$\textnormal{\texttt{init\_case\_3}}(\mathcal{A})$}{
        $H \gets \emptyset$\Comment*[r]{predecessor data structure over $\mathcal{H} = \lbrace \qend : ([\tbeg..\tend),[\qbeg..\qend)]) \in \mathcal{A} \rbrace$}
        $D \gets \emptyset$\Comment*[r]{predecessor data structure over $\mathcal{D} = \lbrace \diag(a) : a \in \mathcal{A} \rbrace$}
        $H.\mathsf{insert}(0 \mapsto a_0)$\Comment*[r]{base case for $a_0 = \startanchor$}
        $D.\mathsf{insert}(0 \mapsto a_0)$\Comment*[r]{base case for $a_0 = \startanchor$}
        $\overline{R} \gets \big(  ((\rightarrow,a_0) \mapsto a_0), ((\searrow,a_0) \mapsto \perp) \big)$\;\Comment*[f]{doubly-linked list of dividing lines (\Cref{obs:dividinglines}) representing $R_t$, where each line is mapped to anchor $a_{R_t[q]}$ of the corresponding region.}\BlankLine
        $Q \gets \big\lbrace t \mapsto \emptyset : t \in \lbrace \tbeg^j : j \in [1..n]\rbrace \cup \lbrace \tend^j : j \in [1..n] \rbrace\big\rbrace$\;\Comment*[f]{Intersection update queues before each $T$-coordinate in the input.}\BlankLine
    }
    \myproc{$\textnormal{\texttt{process\_startpoint\_case\_3}}(a_j)$}{
        Process all intersection updates in $Q[\tbeg^j]$\;
        $a_{\rightarrow} \gets H.\mathsf{pred}(\qbeg^j + 1).\mathsf{value}$\Comment*[r]{closest $\rightarrow$ at or above $\qbeg^j$}
        $a_{\searrow} \gets D.\mathsf{succ}(\diag(a_j) - 1).\mathsf{value}$\Comment*[r]{closest $\searrow$ at or above $\diag(a_j)$}
        \uIf(\Comment*[f]{$(\rightarrow,a_\rightarrow)$ is closer}){$\qbeg^\rightarrow < \tbeg^j - \diag(a_\searrow)$}{%
            $a_i \gets \overline{R}.\mathsf{node}(\rightarrow,a_\rightarrow).\mathsf{value}$\;
        }
        \Else(\Comment*[f]{$(\searrow,a_\searrow)$ is closer}){%
            $a_i \gets \overline{R}.\mathsf{node}(\searrow,a_\searrow).\mathsf{value}$\;
        }
        \Return $\qbeg^j + C[i] - \tend^i$\Comment*[r]{\Cref{eq:casethreefinal}}
    }
    \myproc{$\textnormal{\texttt{process\_endpoint\_case\_3}}(a_i)$}{
        Process all intersection updates in $Q[\tend^i]$\;
        $a_{\rightarrow} \gets H.\mathsf{pred}(\qend^i).\mathsf{value}$\Comment*[r]{closest $\rightarrow$ above $\qend^i$}
        $a_{\searrow} \gets D.\mathsf{succ}(\diag(a_i)).\mathsf{value}$\Comment*[r]{closest $\searrow$ above $\diag(a_i)$}
        \uIf(\Comment*[f]{$(\rightarrow,a_\rightarrow)$ is closer}){$\qbeg^\rightarrow < \tend^i - \diag(a_\searrow)$}{%
            $a_k \gets \overline{R}.\mathsf{node}(\rightarrow,a_\rightarrow).\mathsf{value}$\;
            \If{$C[i] - \tend^i \le C[k] - \tend^k$}{%
                $\overline{R}.\mathsf{node}(\rightarrow,a_\rightarrow).\mathsf{insert}\texttt{\_}\mathsf{next}((\rightarrow,a_i) \mapsto a_i)$\;
                $\overline{R}.\mathsf{node}(\rightarrow,a_i).\mathsf{insert}\texttt{\_}\mathsf{next}((\searrow,a_i) \mapsto a_k)$\;
            }
        }
        \Else(\Comment*[f]{$(\searrow,a_\searrow)$ is closer}){%
            $a_k \gets \overline{R}.\mathsf{node}(\searrow,a_\searrow).\mathsf{value}$\;
            \If{$C[i] - \tend^i \le C[k] - \tend^k$}{%
                $\overline{R}.\mathsf{node}(\searrow,a_\searrow).\mathsf{insert}\texttt{\_}\mathsf{next}((\rightarrow,a_i) \mapsto a_i)$\;
                $\overline{R}.\mathsf{node}(\rightarrow,a_i).\mathsf{insert}\texttt{\_}\mathsf{next}((\searrow,a_i) \mapsto a_k)$\;
            }
        }
        $H.\mathsf{insert}(\diag(a_i) \mapsto a_i)$\Comment*[r]{line $(\rightarrow,a_i)$ becomes active}
        $D.\mathsf{insert}(\diag(a_i) \mapsto a_i)$\Comment*[r]{line $(\searrow,a_i)$ becomes active}
        If $\overline{R}.\mathsf{node}(\rightarrow,a_i).\mathsf{prev}$ is a diagonal line, queue their intersection update at $Q[t]$ with $t$ the closest $T$-coordinate after or including the intersection\;
        If $\overline{R}.\mathsf{node}(\searrow,a_i).\mathsf{next}$ is a horizontal line, queue their intersection update at $Q[t]$ with $t$ the closest $T$-coordinate after or including the intersection\;
    }    \caption{Partial implementation of the handling of the gap-gap, bigger $T$-gap case (\Cref{eq:fourcases}), as part of \Cref{algo:main}, and following Eppstein et al.~\cite{EppsteinGGI92one}. The doubly-linked list $\overline{R}$ is a compact representation of the owner array $R_t$ (\Cref{eq:ownerarray}) of elements $(d,a) \in \lbrace \rightarrow,\searrow \rbrace \times \mathcal{A}$ mapping each dividing line to the originating anchor $a_{R_t[q]}$.
    We omit the handling of intersection updates, and we omit the edge cases in \texttt{process\_endpoint\_case\_3} where $\anchorend(a_i)$ falls on top of an active diagonal or horizontal line.
    }\label{algo:casethree}
\end{algorithm}

\begin{lemma}\label{lem:case3}
    We can compute values $C^3[j]$ (\Cref{eq:casethreefinal}) as part of the main left-to-right line sweep (\Cref{algo:main}) in $O(n \log \log n)$ time and $O(n)$ space.
\end{lemma}
\begin{proof}[Proof sketch]
The correctness follows from \Cref{eq:casethreefinal}, \Cref{obs:intersectionare,obs:dividinglines}, and from the correct handling of the compressed representation of $R_t$ as partially shown in \Cref{algo:casethree} and fully shown in \cite[Section 4]{EppsteinGGI92one}.
As discussed in \cite{EppsteinGGI92one}, the total number of diagonal or horizontal lines is $2n$ and the cost of handling intersections can be charged to the diagonal and horizontal lines of origin.
The only modification to the solution by Eppstein et al.\ is that to reach both $O(n \log \log n)$ time and $O(n)$ space, we delay the intersection updates until we reach the closest column to the right of it that contains an anchor start or end, and we execute these updates in an unsorted, first-in-first-out order.
We argue that this modification maintains correctness, since in a batch of intersection updates we remove areas of influence without changing the order (i.e.\ we get a subsequence of $R_t$ and of the active dividing lines).
We can check for new intersection points at any modification of $R_t$: if they occur before the current $t$ they are handled directly; otherwise they are queued accordingly in $O(\log \log n)$ time per point using a dynamic predecessor query structure on the $T$-positions hit by some anchor startpoint or endpoint.
We leave the full pseudocode implementation and a more detailed proof for a future extended version of this work.
\end{proof}

\begin{figure}[t]
    \centering
    \includegraphics[width=0.8\textwidth]{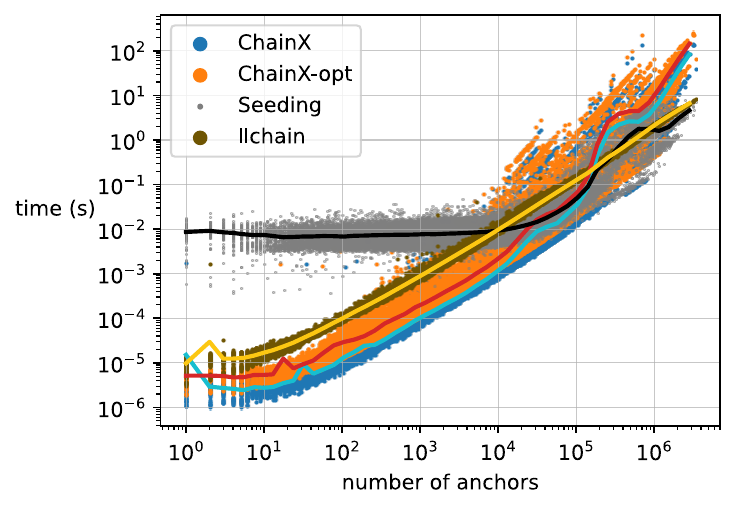}%
    \caption{
      Time spent on seeding (grey, searching MEMs using MUMmer4) and chaining the resulting anchors using
      \texttt{ChainX}, \texttt{ChainX-opt}, and \texttt{llchain} for MEMs of
      length $\geq 50$ between 100\,000 HiFi reads and a human genome reference.
      The highlighted curves indicate the average running time after bucketing
      by the number of anchors.
      }%
    \label{fig:results}
\end{figure}

\section{Experiments}\label{sec:experiments}
We implemented our algorithm in \texttt{llchain} and compare it to
\texttt{ChainX} and \texttt{ChainX-opt}.
Our implementation uses a standard C++ \texttt{std::map} as the predecessor structure and a segment tree for suffix minimum queries, with $O(\log n)$ time operations, and uses the \texttt{std::sort} and \texttt{std::stable\_sort} sorting implementations.
Experiments are run on an Intel Xeon Gold 6148 2.40GHz machine with 376 GB of RAM in the University of Helsinki Kale cluster (the Slurm job is run with options \texttt{{-}{-}exclusive} and \texttt{{-}{-}mem=60G}).

Similar to~\cite{RizzoCM25},
we consider 100\,000 PacBio HiFi reads of a human genome%
\footnote{Run m64004 available at \url{https://s3-us-west-2.amazonaws.com/human-pangenomics/index.html?prefix=T2T/scratch/HG002/sequencing/hifi/}.}
as generated by the T2T project.
We then use MUMmer4 \cite{mummer-1,kurtz2004versatile,mummer-4} to find all
maximal exact matches (MEMs) of length at least 50 between each read and the concatenation of all
chromosomes in the T2T-CHM13 human genome reference \cite{complete-human-genome,chm13v2}.
We exclude from the analysis 216 reads (0.2\%) presenting erroneous inexact matches reported by MUMmer4.\footnote{See \url{https://github.com/mummer4/mummer/issues/235}.}
For each method, we measure how long it takes to
sort the MEMs, double check the absence of overlapping or touching anchors on same diagonal\footnote{Note that \texttt{llchain} can handle an arbitrary anchor set in input, since if the check fails it joins any anchor pair $a,a'$ such that $\connect(a,a') = 0$: the anchored edit distance is unaffected.}, and
then find the best semi-global chain.
We call each method directly via its API.

Results can be seen in \cref{fig:results}.
\texttt{llchain} is around $10\times$ slower than \texttt{ChainX} and $4\times$
slower than \texttt{ChainX-opt} for 100 to 10\,000 anchors, but in this case, the
time needed for seeding dominates. When the number of anchors is large (above
100\,000), \texttt{ChainX} and \texttt{ChainX-opt} have a high variance in their
running time, and likely slow down when there are many matches in repetitive regions.
In these cases, the running time grows roughly quadratically, and the slow
instances dominate the overall running time.
On the other hand, \texttt{llchain} has a very consistent running time and, as
predicted by the theory, has a complexity roughly linear in the number of anchors.
Overall, seeding takes on average
0.044s,
\texttt{ChainX}
0.133s,
\texttt{ChainX-opt}
0.243s,
\texttt{llchain}
0.056s.
Thus, \texttt{llchain} is
only slightly slower than the seeding on average, whereas the largest cases hurt
the other methods, making them over
$2.3\times$
slower.
The exact methods \texttt{ChainX-opt} and \texttt{llchain} always agree on the optimal cost, and consistently with the chaining of maximal unique matches~\cite{RizzoCM25} \texttt{ChainX} returns a suboptimal chain for 2904 reads (2.9\% of the reads considered): for this subset of reads, the median improvement of the anchored edit distance is 1296, and the median improvement ratio is significant, $80\times$.
Separate fuzz tests show that \texttt{llchain} gives the same results as a naive $O(n^2)$-time implementation of \Cref{eq:fourcases}.

\section{Conclusion}

Our new algorithm for anchored edit distance combines many existing ideas and runs in $O(n \log \log n)$ time and $O(n)$ space.
Unlike the original chaining algorithm with this time complexity~\cite{EppsteinGGI92one,DBLP:journals/algorithmica/ChaoM95}, it achieves $O(n)$ space rather than e.g.\ $O(|Q|+|T|)$ space, which might be of independent interest.
While a faithful $O(n \log \log n)$-time implementation is unlikely to be competitive in practice (see e.g.~\cite{navarropredecessor}), we have showed the $O(n \log n)$-time behavior of our simplified \texttt{llchain} in a realistic long-read setting.
We further note that our implementation is not yet optimized, and in particular
that \texttt{std::map} and \texttt{std::list} are known to be slow.
Future versions of \texttt{llchain} will use B-trees or one of the engineered predecessor implementations of
Dinklage et al.~\cite{engineering-predecessor}, which are reported to be
respectively $4\times$ and $16\times$ faster than \texttt{std::map}, and it should be possible to remove the linked list, as it is implicitly represented by the predecessor structures on horizontal and diagonal lines.

Another current limitation of our work is that we do not assess the quality of the anchored edit distance for long reads compared to the other chaining methods, nor we integrate \texttt{llchain} in a seed-chain-extend workflow.
Towards this direction, future work includes extending our method to the gap cost of A*PA \cite{DBLP:journals/bioinformatics/KoerkampI24}, as well as a generalization that handles inversions.
Furthermore, it should be possible to adapt the method to allow for arbitrary fragment scores, affine gap cost penalties, and inexact matches.

\bibliography{bibliography}



\end{document}